\documentclass[11pt]{article}
\usepackage{fullpage}

\usepackage{amsmath,amsfonts,amssymb,amsthm}
\usepackage{mathtools}
\usepackage{commath}
\usepackage[round]{natbib}

\usepackage{graphics}
\usepackage{tensor}
\usepackage{tikz}
\usepackage{qtree}
\usetikzlibrary{trees}

\newtheorem{theorem}{Theorem}[section]

\newtheorem{proposition}[theorem]{Proposition}
\newtheorem{algo}[theorem]{Algorithm}

\newcommand{\real}{\mathbb{R}}

\begin{document}

\linespread{1.5}
\parindent=25pt
\renewcommand{\theequation}{\thesection.\arabic{equation}}
\pagestyle{plain}

\title{Risk-Control Strategies }

\author{Patrice Gaillardetz\thanks{Patrice Gaillardetz, Ph.D.,
is an Associate Professor in Department of Mathematics and
Statistics at Concordia University, Montreal, Quebec H3G 1M8,
Canada, e-mail: patrice.gaillardetz@mathstat.concordia.ca } 
\and Saeb Hachem\thanks{Saeb Hachem, Ph.D. e-mail: saeb@videotron.ca} }

\date{}

\maketitle \noindent\hrulefill

\begin{abstract}

In this paper, we consider the pricing of derivative products that involve dynamic hedging strategies and payments within the planning horizon. Equity-indexed annuities (EIAs), Guaranteed investment certificate (GIC), American and Barrier options are typical examples of these products. Our exploration involves evaluation under different assumptions related to the way the risk is tailored by the issuer. The unified constrained discrete stochastic dynamic programming framework presented in this paper makes use of sequential local minimizing strategies related to stochastic transitions. This sequential minimizations takes into account all intermediate requirements and involves several dynamic risk measures modelling. To demonstrate the flexibility of this framework we present numerical examples featuring GICs and point-to-point EIAs.
\end{abstract}

\noindent\hrulefill\\
Partial hedging; Local risk-minimizing strategies; Stochastic dynamic programming; Linear programming; Parametric linear programming; Risk measures\\
%\noindent JEL Classification: Primary G22, Secondary G32
\newpage

\section{Introduction}

\setcounter{equation}{0}

A unified stochastic dynamic programming framework is proposed for pricing some life insurance and financial products such as the equity-linked products, guaranteed investment certificate (GIC) and other financial contingencies.  These problems are basically valuation of derivative products that involve payments within the planning horizon.  The main goal is to identify dynamic hedging strategies that  control the risk. 

The concept of risk-minimizing was first introduced by \cite{Follmer86}. In spite of its computational simplicity, the square of the loss does not distinguish positive losses from negative ones. A more meaningful approach is to set the hedging strategy using the quantile hedging approach introduced by \cite{Follmer99}; this hedging method is in fact equivalent to the value-at-risk (VaR) risk measure. \cite{Follmer00} extend this idea and construct hedging strategies that minimize positive losses. \cite{Rockafellar} set their hedging strategy by minimizing the conditional-value-at-risk (CVaR) and also by controlling this risk measure in the constraint.

The proposed framework is a flexible alternative to the local risk-minimization approach proposed by \cite{Schweizer88} and \cite{Follmer88}, where they minimize the square of the mismatches process sequentially. \cite{Coleman06} apply this approach to equity-linked products where they also consider options in their hedging portfolio. The local risk-minimizing strategy has been generalized to convex functions in \cite{Abergel11}. \cite{Gaillardetz17} generalize the local risk-minimizing approach using risk measures. They show that their strategies could outperform the quadratic approach.

The optimality of dynamic hedging strategies is taken according to some dynamic portfolio management criteria using different risk measurements. For the issuer agent or company, the use of a common platform and resolution techniques, based of dynamic programming, allows evaluating the relative impact of different dynamic risk management models and pricing different products in actuarial and financial sciences. 

The nuclear of our backward stochastic dynamic models are the optimization problems related to temporary conditional transitions of the multidimensional random process. The control decisions are the corresponding amounts of assets held in the portfolio. The state variables, the objective function, and the constraints are different according to the way the risk is measured and controlled; the modelling complexity increases with the number of assets, the transaction costs, the multiple ways of controlling the risk, etc.

The stochastic setting is assumed to be discrete and described by a tree. Over time, the portfolio updates occur in accordance with the new information related to transitions.  For every transition, some ``controllable" losses are tolerated. Self-financing is hence guaranteed with a high level of confidence but not for all outcomes. Similar to portfolio management problems, this can be done either by:
\begin{itemize}
\item minimizing risk measures with the value of the portfolio as a state constraint;
\item minimizing the value of portfolios under local risk measure constraint. 
\end{itemize}

The latter generalizes the approach proposed by \cite{Gaillardetz17} and rephrases their work in terms of linear programming (LP) optimization. For the conditional value-at-risk, we use the linear modelling of \cite{Rockafellar}.

Among models with risk measure in the objective, we consider three different approaches:
\begin{itemize}
\item minimizing the local risk measure with constraints on future risk measures;
\item minimizing the weighted average between the local risk measure and the future risk measures; 
\item minimizing the coherent dynamic risk measure introduced by \cite{Riedel04}.
\end{itemize}

We will study, theoretically and numerically, the impact on the pricing of these fundamental models and compare their relative performances. We also give some recommendations to calibrate the risk measure parameters. 

One of our main contributions is that the proposed framework allows the involvement of multiple assets in the hedging portfolio. Numerical tests show a significant reduction of the derivative price by adding assets. The proposed flexible algorithms work for both recombining trees and unfolded trees. We also introduce constraints that could reinforce the local risk management control. Through state constraints and variables, the total losses all along paths or scenarios could be controlled.   

The paper is organized as follow. The first two sections introduce the valuation framework and hedging portfolios. Section 4 introduces the optimization algorithms that allow hedging portfolio selections. It presents several technics where the risk measure is used in the constraint and in the objective function. Section 5 concludes the article with the numerical implementations.

\section{Financial Framework }

In this paper, to proceed with numerical calculations, we assume that the input random process is a multidimensional discrete process with a finite number of realizations. Depicting this process as an event tree is convenient for stochastic dynamic optimization; the nodes of the tree represent the possible values of the random process and the arcs the transitions, i.e. possible continuations of history. 
If the random process is Markovian, modelling the discrete process as a recombining tree leads to a significant reduction in the computational complexity, by comparison to represent the process as an unfolded tree or a scenario tree. The unfolded tree is required to model more general stochastic processes where the evolution of the underlying quantities is strongly path-dependent. The main difference between these two categories of trees lies in the fact that in a recombining tree many arcs could lead to a node, which implies many possible histories, while only one arc leads to a node in an unfolded tree, to translate the long dependence on history.  Figure \ref{f_tree} shows both trees over 3 periods or 3 transitions.

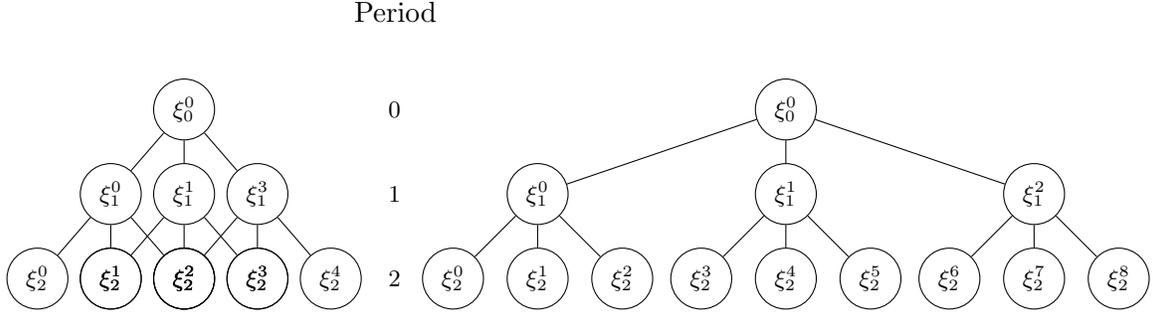
\begin{figure}

\hspace{5.1cm} Period
\\

\hspace{-2cm} 
\begin{tikzpicture}[scale=0.75,font=\footnotesize, rotate=-90]

  \tikzstyle{solid node}=[circle,draw,inner sep=1.5,fill=black]
  \tikzstyle{hollow node}=[circle,draw,inner sep=1.5]
  \tikzstyle{level 1}=[level distance=15mm,sibling distance=13mm]
  \tikzstyle{level 2}=[level distance=15mm,sibling distance=13mm]
  \tikzstyle{level 3}=[level distance=15mm,sibling distance=13mm]
  \vspace{15mm}
  \node[circle,draw]{$\xi^0_0$}[grow=east]
    child{node[circle,draw]{$\xi^0_1$}
      child{node[circle,draw]{$\xi^0_2$} edge from parent node[left]{}}
	 child{node[circle,draw]{$\xi^1_2$} edge from parent node[left]{}}
      child{node[circle,draw]{$\xi^2_2$} edge from parent node[left]{}}
             %     edge from parent node[left,xshift=-5]{$A$}
    }
    child{node[circle,draw]{$\xi^1_1$}
      child{node[circle,draw]{$\xi^1_2$} }
	 child{node[circle,draw]{$\xi^2_2$} edge from parent node[left]{}}
     child{node[circle,draw]{$\xi^3_2$} edge from parent node[left]{}}
             %     edge from parent node[left,xshift=-5]{$A$}
    }
    child{node[circle,draw]{$\xi^3_1$}
      child{node[circle,draw]{$\xi^2_2$} }
        child{node[circle,draw]{$\xi^3_2$} edge from parent node[left]{}}
    child{node[circle,draw]{$\xi^4_2$} edge from parent node[right]{}}
 %     edge from parent node[right,xshift=5]{$B$}
    };
%\end{tikzpicture}

 %\caption{Recombining tree.}
 %\label{fl}
%\end{figure}

%\begin{figure}[h!]
%\begin{tikzpicture}[scale=0.75,font=\footnotesize,rotate=-90]

\hspace*{2.8cm} 

%\begin{tikzpicture}[scale=0.75,font=\footnotesize, rotate=-90]
  \tikzstyle{solid node}=[circle,draw,inner sep=1.5,fill=black]
  \tikzstyle{hollow node}=[circle,draw,inner sep=1.5]
  \tikzstyle{level 1}=[level distance=15mm]
  \tikzstyle{level 2}=[level distance=15mm]
  \tikzstyle{level 3}=[level distance=15mm]
  \node{0}[grow=east]
    child{node{$1$} edge from parent[draw=none]
      child{node{$2$} edge from parent[draw=none] }};

\hspace*{5.2cm}  
  \tikzstyle{solid node}=[circle,draw,inner sep=1.5,fill=black]
  \tikzstyle{hollow node}=[circle,draw,inner sep=1.5]
  \tikzstyle{level 1}=[level distance=15mm,sibling distance=4.4 cm]
  \tikzstyle{level 2}=[level distance=15mm,sibling distance=1.5cm]
  \tikzstyle{level 3}=[level distance=15mm,sibling distance=5mm]
  \node(0)[circle,draw]{$\xi^0_0$}[grow=east]
    child{node[circle,draw]{$\xi^0_1$}
      child{node[circle,draw]{$\xi^0_2$} edge from parent node[left]{}}
	 child{node[circle,draw]{$\xi^1_2$} edge from parent node[left]{}}
      child{node[circle,draw]{$\xi^2_2$} edge from parent node[left]{}}
             %     edge from parent node[left,xshift=-5]{$A$}
    }
    child{node[circle,draw]{$\xi^1_1$}
      child{node[circle,draw]{$\xi^3_2$} }
	 child{node[circle,draw]{$\xi^4_2$} edge from parent node[left]{}}
     child{node[circle,draw]{$\xi^5_2$} edge from parent node[left]{}}
             %     edge from parent node[left,xshift=-5]{$A$}
    }
    child{node[circle,draw]{$\xi^2_1$}
      child{node[circle,draw]{$\xi^6_2$} }
        child{node[circle,draw]{$\xi^7_2$} edge from parent node[left]{}}
    child{node[circle,draw]{$\xi^8_2$} edge from parent node[right]{}}
 %     edge from parent node[right,xshift=5]{$B$}
    };
\end{tikzpicture}
 \caption{Recombining and unfolded trinomial trees.}
 \label{f_tree}
\end{figure}

Lattice models have been intensively used to describe stocks, stock indices, interest rates, and other financial securities due to their flexibility and tractability. They also arise naturally when dealing with holders' lives since they are described using counting processes. 

As for our algorithm, the computational efforts are proportionate to the number of nodes, the running time grows linearly with the number of periods in the case of a recombining tree. Otherwise, it grows exponentially and faces the usual curse of dimensionality. 

Although this is not necessary for our framework, we assume for the sake of notational simplicity that the planning horizon is cut off into successive dates, $t=0,1,\cdots, T$, and the values of the random process, $\xi_0$, $\xi _1$, ..., $\xi _T$ are successively known at this sequence of dates. Hence, for a given date, whatever the value of the random process is, the duration of any next transition is constant (an alternative way to building tree is to use time-varying transition lengths). These assumptions translate into a tree having a finite number of nodes organized in a finite number of levels associated to dates and the transitioning represented by arcs is made from one level to the next (see Figure \ref{f_tree}).

Let $i_t$ be the node index number $i$ of the period $t$ and $\xi^i_t$ the value of the stochastic process relative to the node $i_t$. For $t=0$, the unique node, or the root node, is associated with the known value of the stochastic process $\xi^0_0$. At time $t=1$, we have as many nodes as possible values of $\xi _1$. To mark the continuation of history, each of these nodes is connected by an arc to the root node. Generally, at level $t, t>0$, the set of nodes $I_t$ includes all possible values of $\xi _t$, each node $i$ is connected to at least one node $j$ of level $t-1$. For a non-recombining tree, this connection is unique. %Ajuster les arbres aver la notation
%, the value of the process at node $i$ is $\xi^i_t$

Let $I_t|i_{t-1}$ denote a set of nodes including all the realization of $\xi _t$ given parent node $i_{t-1}$. Let $p_{i_t|i_{t-1}}$ be the conditional probability of moving from the node $i_{t-1}$ to  $i_t$ and $ E_{i_{t-1}}[.]$ the associated conditional expected value operator, that is
\[
E_{i_{t-1}} [f(\xi_t)]=\sum_{i_t \in I_t|i_{t-1}} f(\xi^{i}_t)p_{i_t|i_{t-1}}, 
\]
where $f$ is un function of the random variable $\xi_t$.

$\xi^i _t$ is, in general, a multivariate data vector. For guaranteed investment certificate (GIC) the vector is a unidimensional data vector when the contacts depend only on the stock price. If the contracts depend on multiple assets or stochastic interest rate, the vector is multi dimensional. For the equity-indexed annuities (EIAs) problem, the data vector needs to add one dimension for the policyholders' cohort. In both cases, the dimension related to the random stock price is denoted by $S_{i_t}$, which is the stock value at the node $i_t$.  

Although the proposed framework can easily handle stochastic interest rates,  we assume for the sake of simplicity that the interest rate is deterministic. Let $r$  be the force of interest, i.e. $r$ is a nominal rate of interest compounded continuously.

\section{Hedging Portfolio and Loss Function}   
\label{hplf}

%The stock process that represents the value of the stock at time $t$ is denoted by $S(t), t = 0, 1, 2, \cdots$, which is real-valued functions where $S(0)$ is the initial level of the stock. The model assumes that the stock could be traded in each period and the usual frictionless financial market: no tax, no transaction costs, etc.  

%????????????????????????????????
For a given transition, the stochastic process losses are defined as follow. Let $x_{i_t}$ be a vector of control variables, which are assets composing the hedging strategy held by the issuer for the outcome $i_t$. Let $F_{i_{t}}(x_{i_t})$ denote the function giving the value of an investment strategy involving $x_{i_t}$. If for example,  $x_{i_t}$ is equal to $(a_{i_t}, b_{i_t}, c_{i_t} )$ where $ a_{i_t}$, $b_{i_t}$, and $c_{i_t}$ are respectively the amounts of stocks, cash, and European call options held by the issuer, then $F_{i_{t}}(x_{i_t})=a_{i_t}+b_{i_t}+c_{i_t}$. The sets of investments could have more components such as other European options or multiples risky assets. All along this paper for modelling and illustrative purposes we consider this three dimensional set of investment.

 %For example, an American call option exercise price is the difference between stock and the strike price if this expression is positive, otherwise it is equal to 0. 
 
Let $C_{i_t}$ the continuation value, which is the amount required to pursue the issuer's operations beyond $i_t$, which will be specified for each model. Depending on the dynamic optimization problem, $C_{i_t}$ could be function depending on state variables. Let $P_{i_t}$ the guaranteed payable benefit relative to node ${i_t}$ and $G_{i_{t}}(P_{i_t} , C_{i_t})$ summarizes the needed amount required at node $i_t$. In the case of GICs, the benefit is paid  only at maturity ($T$) and for all intermediate periods $G_{i_{t}}(P_{i_t} , C_{i_t})= C_{i_t}$. In the case of EIAs, $G_{i_t}$ is the sum of the death benefits and the continuation value for the survivors, $G_{i_{t}}(P_{i_t} , C_{i_t})=P_{i_t} + C_{i_t}$.

Let $z_{i_t}$ denote the state vector relative to node $i_t$ and $H_{i_{t}}( z_{i_{t-1}},z_{i_t} , x_{i_{t-1}},\xi^i_t, G_{i_t} )$ is a multi-function that describes the dynamic states evolution. For example, it could be used to introduce transaction costs on assets. In the case of transaction fees on an asset, we have to distinguish between the stock already held $a_{i_{t-1}}$ and the variation $\delta_ {a_{i_{t-1}}}$ since the fees affect the variation only. Hence, $a_{i_{t-1}}$ becomes a component of the sate vector $z_{i_{t-1}}$ and  $\delta_ {a_{i_{t-1}}}$ is a component of the control vector $x_{i_{t-1}}$. The state equation related to the state variable $a_{i_{t-1}}$ is $ a_{i_t} = a_{i_{t-1}} + \delta_{a_{i_t}}$. If the hedge portfolio includes many assets with transaction fees,  many state variables and similar equations are required. In this case, the value of the investment strategy needs to be adjusted to take into account the variation $\delta_{a_{i_t}}$.

% Note that most of the functions (e.g. $L$, $W$, $C$, etc) could depend on the state vector $z_{i_t}$ if it is introduced in the algorithm. 
 
$W_{i_{t-1}} ( x_{i_{t-1}},  \xi^i_t,z_{i_{t-1}}) $ is the accumulation of hedge portfolio relative to the transition $i_t|i_{t-1}$ prior to the payment at time $t$. In other words, $W_{i_{t-1}}$ represents the value at the end of the period of the replicating portfolio set at node $i_{t-1}$.  Note that $W_{i_{t-1}}$ is a random variable function of the different outcomes $\xi^i_t$. That is for the three dimensional set of investment
\begin{align}
W_{i_{t-1}} ( x_{i_{t-1}},  \xi^i_t,z_{i_{t-1}})=a_{i_{t-1}}\frac{S_{i_t}}{S_{i_{t-1}}}+b_{i_{t-1}}e^{r}+c_{i_{t-1}}\frac{O_{i_t}(T)}{O_{i_{t-1}}(T)},\label{wtm1}
\end{align}
where $O_{i_{t-1}}(T)$ is the price of the European call option with maturity $T\geq t$ at node $i_{t-1}$. This accumulation hedge portfolio is linear with respect to the investment variables.

%of the issuer's future obligations

%For the node $i_t $, the required value is a function of the guaranteed payment benefit $P_{i_t}$ and the future capital requirement $ V_{i_t}(z_{i_t}, x_{i_t}) $.  ???
For the transition $i_t| i_{t-1}$, the issuer has a surplus if the value of investments or hedge portfolio prior to the payment is greater or equal to the required value, i.e.
\begin{align} 
 W_{i_{t-1}} ( x_{i_{t-1}}, \xi^i_t,z_{i_{t-1}} )   \geq   G_{i_{t}}(P_{i_t} , C_{i_t}).\label{WG}
\end{align}
If this inequality  is observed for all nodes, the hedging strategy refers to the super-replicating strategy. Otherwise if for some transitions this inequality is violated the issuer incurs temporary losses which, in this context, are simply the difference between the amounts required less the accumulation of the hedge portfolio.  For the transition $i_t|i_{t-1}$, the discounted loss random variable is
\begin{align}
 L_{i_{t-1}} (x_{i_{t-1}}, \xi^i_t,  G_{i_t},z_{i_{t-1}})  =  G_{i_t} (P_{i_t} , C_{i_{t}})  - W_{i_{t-1}} ( x_{i_{t-1}}, \xi^i_t,z_{i_{t-1}} ).\label{ltm1}
\end{align}
 
\begin{proposition} If $W_{i_{t-1}} $ is concave (or linear) and $G_{i_{t}} $ is convex then the loss function $L_{i_{t-1}} $ is convex. In addition if $G_{i_{t}} $ and $W_{i_{t-1}} $ are piecewise linear then $L_{i_{t-1}} $ is piecewise linear. 
 \end{proposition}

\begin{proof}
If $W_{i_{t-1}}$ is concave, $-W_{i_{t-1}}$ is convex, and the sum of two convex functions is convex (see \cite{Rockafellar70}).

 The domain of the sum function is defined by the intersection of the domains of $G$ and W which are piecewise linear on this domain. For a given piece(polyhedral) of $G$, there is a polyhedral covering from the domains of pieces (possibly one) of $W$. Since each part of this polyhedral covering is linear, the  sum function is linear and hence the sum function is piecewise linear on this given piece of $G$. Similarly, the sum function is piecewise linear on the domain of every other piece of $G$. As the number of pieces of $G$ is finite, the sum function is piecewise linear with a finite number of pieces. It is continuous since it is convex.% Similarly the sum of two piecewise linear functions is a piece piecewise linear function.
\end{proof}

\section{Hedging Portfolio Selection}
\label{HS}

\cite{Schweizer88} propose a local risk-minimizing strategy that sequentially minimizes the square of the error process. \cite{Gaillardetz17} propose partial hedging strategies that allows some positive losses by containing risk measure to be smaller than a given threshold. These risk-control strategies can be generalized by including more constraints and variables and using linear programming technics.

All the proposed optimization algorithms set the hedging portfolio using backward dynamic programming approach. First, the optimization allowing to define the hedge portfolio is performed for all $i_{T-1}$. Based on these results, the hedge portfolio can be obtained for all possible outcomes $i_{T-2}$. This backward process leads to the initial value of the hedge portfolio. 

Let $V_{i_t}(z_{i_t})$ denote the backward cost-to-go function of the dynamic programming problem at node $i_t$.

\subsection{Probabilistic or Chance Constraint}

The first approach chooses the hedging portfolio based on probabilistic constraints. It allows controlling the loss function with some level of confidence. The probabilistic constraint guarantees that the losses are controlled with a probability of $c$. Hence, the optimization problem is given by
\begin{algo}\label{algo1}
For all $t=T, T-1,  \cdots, 1$ and all $i_{t-1  }\in I_{t-1}$,  
\begin{align}
V_{i_{t-1}} (z_{i_{t-1}})= \min_{ x_{i_{t-1}},z_{i_t} }   F( x_{i_{t-1}})  
 \end{align}
under the constraints
\begin{align}
\Pr{_{i_{t-1}}} [L_{i_{t-1}} (x_{i_{t-1}}, ~\xi^i_t,  G_{i_t}, z_{i_{t-1}}) \leq  \gamma_0)] \geq & c \label{pc1} \\  
H_{i_t}( z_{i_{t-1}}, z_{i_t},x_{i_{t-1}},\xi^i_t, G_{i_t}  )  =&0 \label{pc2}
\end{align}
where $V_{i_T}(z_{i_T})=P_{i_T}$ for all $i_T$,  and $\Pr{_{i_{t-1}}}[.]$ is the probability given $i_{t-1}$.
\end{algo}
This limits the probability of losses larger than a threshold parameter $\gamma_0$ to $1-c$.

\cite{charnes} are the first to defined disjoint chance constraint programs. \cite{Miller} study joint chance constraint program. This later work was generalized to non-independent random variables by \cite{Prekopa73}. 

Probabilistic constrained programs with discrete probability distributions can be reformulated as mixed-integer programs (\cite{Prekopa}, \cite{Ruszczynski02}, \cite{Luedtke}, and \cite{Kucukyavuz}). Mixed-integer programs are way more complex to solve than linear programs and  their cost-to-go function are generally non-convex and discontinuous. 

%The constraint (\ref{pc1}) can be defined continuous probability distributions in (\ref{pc1}), the programs can be convex or non-convex according to the distribution property (\cite{Prekopa} and ). In the most general case we use sampling and approximations of the probabilistic constraint. These methodologies are difficult to solve with parametric formulations. 

A major drawback of the probabilistic approach is that the level of the losses is not penalized and no control is imposed on these violations. Risk measures could remedy to this unpenalized violation.

\subsection{Risk Measure as Constraint}
\label{rmc}
Risk measures have been widely used by financial institutions such as insurance and investments companies to evaluate the risk level of business lines. This widespread use is mainly due to its meaningfulness in a business setting. Hence, the probabilistic constraint in the Algorithm \ref{algo1} is replaced by a risk measure. That is
\begin{algo}\label{algo2}
For all $t=T, T-1,  \cdots, 1$ and all $i_{t-1}\in I_{t-1}$,  
\begin{align}
V_{i_{t-1}} (z_{i_{t-1}})= \min_{ x_{i_{t-1}} ,z_{i_t} }   F( x_{i_{t-1}})   \label{ro}
\end{align}
Under the constraints
\begin{align}
\rho_c (L_{i_{t-1}} (x_{i_{t-1}}, ~\xi^i_t,  G_ {i_t}, z_{i_{t-1}}) )  \leq &\gamma_0  \label{rc1} \\
H_{i_t}( z_{i_{t-1}}, z_{i_t},x_{i_{t-1}},\xi^i_t , G_{i_t} )  =&0 \label{rc2}
\end{align}
where $\rho_c(.)$ is the risk measure that quantifies the total risk exposure at risk level $c$.
\end{algo}

\begin{proposition} If $\rho_c$, $F_{i_{t-1}} $ and $L_{i_{t-1}} $ are convex and $H_{i_{t}} $ is linear then the optimization problem is convex. If in addition $F_{i_{t-1}} $ and $L_{i_{t-1}} $ are linear and $\rho_c$ is piecewise linear then the optimization problem is convex piecewise linear.  
 \end{proposition}
 
  \begin{proof}
Since $\rho_c$ is convex and the composition of two convex functions $\rho_c (L_{i_{t-1}} (.))$ is convex then the problem is convex.  The rest of the proof is obvious.
 \end{proof}
In the following, we introduce several risk measures (including the CVaR) that are piecewise linear.
This generalizes the approach proposed by \cite{Gaillardetz17}, where the portfolio composition can be extended to multiple assets and the possibility to use state variables.

In the case there are no state variables, the backward cost-to-go function is a constant and nothing else than the continuation value. The objective minimization of Algorithm \ref{algo2} is simply
\[
C_{i_{t-1}} = \min_{a_{i_{t-1}},b_{i_{t-1}},c_{i_{t-1}}  }   a_{i_{t-1}}+b_{i_{t-1}}+c_{i_{t-1}}. 
\]
There is no need for the state constraint (\ref{rc2}) and the constraint (\ref{rc1}) becomes $\rho_c (W_{i_{t-1}}- G_{i_t})  \leq \gamma_0 $ where the needed amount required $G_{i_t}$ is constant and $W_{i_{t-1}}$ is the linear function given by (\ref{wtm1}).

The optimization problem (\ref{ro})-(\ref{rc2}) requires different solutions depending on the specific risk measure. The most common and well-known risk measure is the Value-at-Risk (VaR). VaR is widely used due to its ease of calculation. However, incorporating the VaR in optimization problems is a challenge. The Value-at-Risk is defined as the $c$-quantile of the discounted loss function $L_{i_{t-1}}$. That is
\begin{align}
\textnormal{VaR}_c(L_{i_{t-1}}) = \inf \{ y \in \real : \Pr{_{i_{t-1}}}[L_{i_{t-1}}>y] \leq 1-c \}.
\end{align}
for $c\in[0,1]$.

The value-at-risk shares the same drawbacks as the probabilistic approach and requires some approximations to be tractable (see Chapter 4 of \cite{Aharon}) 
 
Conditional Value-at-Risk (CVaR) has been lauded as a more meaningful and appropriate coherent risk measure (see \cite{Artzner99}). Essentially, a coherent risk measure is said to possess the properties of monotonicity, sub-additivity, positive homogeneity, and translation invariance. Note that VaR does not possess the sub-additivity property and is therefore not a coherent risk measure.

The CVaR risk measure, which represents the expected value of the worst $(1-c)$ losses located at the tail delimited by VaR$_c$.

%, is given by
%\begin{align}
%\label{eqn:cvar1}
%\textnormal{CVaR}_c(L_{i_{t-1}}) =(1-c)\int_{L_{i_{t-1}}\geq \pi_c}  L% E_{i_{t-1}} [ L_{i_{t-1}}|L_{i_{t-1}}>\textnormal{VaR}_c(L_{i_{t-1}}) ].???
%\end{align}
%where $\pi_c$ is the $\textnormal{VaR}_c$.
% where $E{_{i_{t-1}}}[.]$ is the expected value with respect to $\xi^i_t$ given $i_{t-1}$.
 
According to \cite{Rockafellar}, CVaR is leading to a tractable optimization convex problem. In the case $\xi^i_t$ is discrete, the constraint (\ref{rc1}) becomes
\begin{align}
-L_{i_{t-1}} + \pi_c^+ - \pi_c^- + u_{i_t} &\geq 0,~~  \forall   i_t  \in I_t|i_{t-1} , \label{cvarc1}\\
u_{i_t} & \geq 0,~~ \forall   i_t  \in I_t|i_{t-1} , \label{cvarc1b}\\
\pi_c^-,\pi_c^+ &\geq 0,\label{cvarc1c} \\
\pi_c^+ - \pi_c^- + \sum_{i_t} \dfrac{u_{i_t} p{_{i_t|i_{t-1}}}  }{1-c} & \leq \gamma_0,\label{cvarc2}
\end{align}
where at optimality $\pi_c^+-\pi_c^-$ is VaR$_c$, which is split in two non-negative sub-variables for the need of standard linear programming solvers. 

Convex piecewise linear problems can be solved using standard linear programming software.

If $L_{i_{t-1}}$ is linear or convex piecewise linear, the constraints (\ref{cvarc1})-(\ref{cvarc2}) are linear or equivalent to a system of linear constraints.

Another risk measure that could be considered is the expected value of the downside risk, which limits the positive losses. In this case, the constraint (\ref{rc1}) becomes
\begin{align}
-L_{i_{t-1}} + u_{i_t} &\geq 0,~~  \forall   i_t  \in I_t|i_{t-1} , \label{cedr1}\\
u_{i_t} & \geq 0,~~ \forall   i_t  \in I_t|i_{t-1} , \label{cedr1b}\\
 \sum_{i_t} u_{i_t} p{_{i_t|i_{t-1}}} & \leq \gamma_0,\label{cedr2}
\end{align}
This algorithm limited the expected value of the positive loss to be smaller than $\gamma_0$. Setting both $\pi_c^+$ and  $\pi_c^-$ to 0 in (\ref{cvarc1})-(\ref{cvarc2})
implies that the CVaR optimization system becomes similar to limiting the downside risk.

The CVaR and the expected value of the downside risk penalize the value of the loss by controlling the violations using the expected value. This can be refined, in a similar way to the portfolio risk management (\cite{Markowitz}), by penalizing the large deviations. For example, it could be achieved by adding a norm constraint 
\begin{align}
\norm{u}_p \leq \gamma_1, \label{cnormp}
\end{align}
 where $\norm{.}_p$ is the norm in $\ell_p$.
 
Different norms can be used to limit the positive losses and they act differently depending on the discretion of the losses. For instance, the norm $\ell_\infty$ limits each component in a box, $\ell_1$ limits the total losses by a linear constraint. The norm $\ell_2$ introduces a quadratic constraint that penalized large deviations. A good alternative that will allow to use linear programming is to imitate the quadratic function using a piecewise convex linear function. For this, each $u_i$ is split into non-negative sub-variables $u_{i_t,j}$ $(j=1,...,J)$ and allocates slopes $f_{i_t,j}$  for each sub-variable, where $J$ is the number of pieces. In this case, the constraint involving the norm (\ref{cnormp}) becomes
\begin{align}
\sum_{ i_t  \in I_t|i_{t-1} }\sum_{j=1}^{J}f_{i_t,j}u_{i_t,j}\leq\gamma_2,\label{cnorm2}
\end{align}
with $u_{i_t}=\sum_{j=1}^{J}u_{i_t,j}.$

Another interesting alternative approach to control positive losses is through path wise constraint. This is done by introducing in the optimization system between $T-1$ and $T$ the following constraint 
 \begin{align}
z_{i_T}\leq \gamma_3 \label{eiT},
\end{align}
 and state equations
 \begin{align}
z_{i_t}=e^r z_{i_{t-1}}+L_{i_{t-1}} \label{eit},
\end{align}
for all $i_t$, where the state variable $z_{i_t}$ is the accumulated losses along a path up to $i_t$. Note that this approach can also be combined with the norm approach.

The special case where $\gamma_3=0$ in (\ref{eit}) leads to a risk-free strategy. Indeed, the issuer could encounter some positive losses, but should receive positive gains that will annihilate the strategy cost. This strategy shares the riskless property of the super-replicating strategy while being less restrictive allowing some positives losses along the contract life. 

% The CVaR is coherent dynamic risk measure  (see ). A dynamic risk measure is. The risk measure is used in the constraint not the objective function.
Since the VaR introduces complexity, we focus on the CVaR. Note that letting $c$ goes to 1 and $\gamma_0$ to 0 leads to the super-replicating strategy.

\begin{proposition} If the optimization problem is linear and $z_{i_{t-1}}$ is in the right hand side, then the cost-to-go function $V_{i_{t-1}}(z_{i_{t-1}})$ is convex piecewise linear.  \label{Propo3}
 \end{proposition}
 
  \begin{proof}
The perturbation of the right hand side convex optimization problem leads to convex function (see \cite{Rockafellar70}). The piecewise linearity results from the simplex algorithm feasibility criteria (see  \cite{Bereanu}).
 \end{proof}
 
To solve a linear program, the simplex algorithm makes many pivot steps.  A multi-parametric (mpLP) is an LP with parameters affecting its coefficients.  An algorithm for mpLP with parameters $z_{i_{t-1}}$ affecting the right hand side solves:
\begin{itemize}
\item First the dual LP relative to a given value of the parameters and generates by the optimal simplex feasibility conditions the parametric feasible polyhedral region which is  the domain of one piece of  $V( z_{i_{t-1}})$;  
\item	Then, in a similar way,  generates a neighbouring  feasible region by  a single dual pivoting step to update the inverse of neighbouring bases (Theorem 1 in \cite{Foote80});
\item	Keep doing the same until covering all the set of parameters.
\end{itemize}
Hence, if the number of pieces of the cost-to-go function is limited, the complexity of solving the mpLP is comparable to solve an LP, since it takes only some extra pivoting steps to solve mpLP from the initial LP. Consequently, by relying on mpLP algorithms we avoid hitting the curse of dimensionality even if we increase the dimension of the state variable $z_{i_{t-1}}$.

All this is fully true if the dual LP is not degenerate. Degeneracy is caused by redundant constraint(s), i.e. more than necessary constraints to define a vertex. In this case, multiple bases are associated to a vertex and the inverse base change by single pivot step does not cause the iteration to follow an edge of the dual polyhedral. 

In non-degenerate case, multi-parametric algorithms explore part of the graph of different bases (or vertices) by simple pivoting steps until the whole set of parameters is covered. The set of neighbouring polyhedral feasible parameters overlap only on the edges.  In the presence of degenerate vertex, a primal problem for mpLP is to select a subset of the bases that describe the vertex to avoid complex feasible parameters region overlapping.

\cite{Gal72} developed the first algorithm to solve mpLP. See \cite{Jonesa07} and \cite{Borrelli03} and the references there for more recent algorithms. Roughly speaking, mpLP algorithms differ by the way they handle degenerated vertexes.

%To solve a linear program, the simplex algorithm makes many pivot steps (the number increases very fast with the problem size). According to Theorem 1 in \cite{Foote} the generation of a neighbouring piece from a given piece is the result of one dual simplex pivot. Hence  the computational complexity of calculating the cost-to-go function $V_{i_{t-1}}(z_{i_{t-1}})$ is comparable to solve a linear program since the number of pieces is small. Consequently by relaying on the parametric linear programming and linear stochastic framework  we avoid hitting the curse of dimensionality even if we increase the number of state variables. Saeb???

\subsection{Risk Measure as Objective}

Similarly to portfolio selection problems, risk measures are used in the objective function to select hedging strategies.  Models based on stochastic programming and coherent risk measure are implemented. 

\subsubsection{Stochastic Programming Approach}

The first algorithm uses the concept of stochastic dynamic programming, where the future risk is averaged. In this approach, the objective function is a compromise between the current CVaR and the expected future CVaR. 
\begin{algo}\label{algo5}
For all $t=T, T-1,  \cdots, 1$ and all $i_{t-1}\in I_{t-1}$,  
\begin{align}
V_{i_{t-1}}(z_{i_{t-1}})= \min_{ x_{i_{t-1}},z_{i_{t}} }   e^{-r}\{ \textnormal{CVaR}_c( L_{i_{t-1}} (x_{i_{t-1}}, ~\xi^i_t,  G_{i_t}, z_{i_{t-1}})+\lambda E_{i_{t-1}}[V_{i_t}(z_{i_t})]\} \label{spa}
\end{align}
Under the constraint
\begin{align}
H_{i_t}( z_{i_{t-1}},z_{i_t},x_{i_{t-1}},\xi^i_t , G_{i_t} )=0,
\end{align}
where $V_{i_T}(z_{i_T})=0$ for $i_T$.
\end{algo}
The objective function takes into account the sum of the local risk measure and  weighted average of the near future risk measure.

%Note that again at least one of the state variables has to be the initial capital needed for financing the investment strategy. 

In this algorithm, the state variable $z_{i_t}$ is the continuation value (i.e. $C_{i_{t}}(z_{i_{t}})=z_{i_t}$) since it is the amount required to continue the issuer's operations beyond $i_t$. 
 
In the case we consider the hedging strategy involving three assets according to \cite{Rockafellar} the Algorithm \ref{algo5} becomes
\begin{align}
V_{i_{t-1}}(z_{i_{t-1}})=e^{-r} \min_{x_{i_t},\pi_c^+,\pi_c^-, u_{i_{t}},\theta^+_{i_{t}},\theta^-_{i_{t}}  }   \pi_c^+ - \pi_c^- + \sum_{ i_t  \in I_t|i_{t-1}} \dfrac{u_{i_t} p{_{i_t|i_{t-1}}}  }{1-c} +\lambda \sum_{ i_t  \in I_t|i_{t-1}} (\theta^+_{i_t}-\theta^-_{i_t}) p{_{i_t|i_{t-1}}}  , \label{spab}
\end{align}piecewise
under the constraints
\begin{align}
-L_{i_{t-1}} (x_{i_{t-1}}, ~\xi^i_t,  G_{i_t}, z_{i_{t-1}}) + \pi_c^+ - \pi_c^- + u_{i_t} &\geq 0,~~  \forall   i_t  \in I_t|i_{t-1} ,\nonumber \\
u_{i_t} & \geq 0,~~ \forall   i_t  \in I_t|i_{t-1} ,\nonumber \\
\theta^+_{i_t}-\theta^-_{i_t}-V_{i_t}(z_{i_t})& \geq 0,~~ \forall   i_t  \in I_t|i_{t-1}, \label{contthetaV}\\
a_{i_{t-1}}+b_{i_{t-1}}+c_{i_{t-1}}&=z_{i_{t-1}},\nonumber
\end{align}
for $t=1,...,T-1$ and where $L_{i_{t-1}}$ is given by (\ref{wtm1}) and (\ref{ltm1}). At the optimal,  the value $\theta^+_{i_t}-\theta^-_{i_t}$ is the next period objective function, which is split in 2 non-negative sub-variables since it could be negative. For $t=T$, we need to remove the constraint (\ref{contthetaV}) and the last term of (\ref{spab}).
%Steps similar to (\ref{rmo1}) and (\ref{cvarc12})-(\ref{cvarc13}) need to be performed for the implementation.
%For this algorithm, at least one additional constraint on positive losses is required (\ref{cnormp})  and/or (\ref{eiT}). Otherwise, our numerical results showed that the objective function CVaR can take unacceptable high values.
%Note that setting $\lambda=0$ in Algorithms (\ref{algo4}) and (\ref{algo5}) leads to Algorithm (\ref{algo3}).

If $V_{i_t}$ is a convex piecewise linear function defined by its supporting hyperplane, the constraint $\theta^+_{i_t}-\theta^-_{i_t}-V_{i_t}(z_{i_t}) \geq 0$ is equivalent to several linear constraints where each of $V_{i_t}(z_{i_t})$ supporting hyperplane is smaller or equal to $\theta^+_{i_t}-\theta^-_{i_t}$. 

Since the function $G_{i_t}$ is linear, the above optimization problem is a parametric linear program  with parameter $z_{i_t}$ in the right hand side.  

\begin{proposition} 
For all $t=T, T-1,  \cdots, 1$ and all $i_{t-1}\in I_{t-1}$, the function $V_{i_{t-1}}(z_{i_{t-1}})$ is a convex piecewise linear function.  
\end{proposition} 

 \begin{proof}
For all $i_{T-1}$, the linear problem involves the state variable $z_{i_{T-1}}$ in the right hand-side. Hence, the optimization problem is a linear parametric ($z_{i_{T-1}}$)  program and $V_{i_{T-1}}(z_{i_{T-1}}) $ is a convex piecewise linear function (see Proposition \ref{Propo3}). For all tree nodes of  ${T-2}$, the cost-to-go function $V_{i_{T-1}}(z_{i_{T-1}})$ in the constraint (\ref{contthetaV}) can be replaced by its supporting hyperplane as mentioned above. Hence, the resulting optimization problem is again a linear parametric program with the parameter affecting the right hand side and the cost-to-go function is convex piecewise linear. By proceeding backward, we can prove that the cost-to-go function is convex piecewise linear using similar arguments.
 \end{proof}
 
We refer to the discussion following the Proposition \ref{Propo3} for an efficient algorithm to build this cost-to-go function.

 %que sur la complete  construction de la fonction linear pas morceau pivotage simplex temps de calcul est comparable a programme linaire Ainsi on peut traiter complexitŽ programmation dynamique n'explose pas avec quelques variables d'etats et arbre combine

%It is important to point out that using the risk measure as objective solves the problem for different hedging portfolios. Hence, the issuer has the liberty of choosing a specific dynamic strategy among the admissible solutions. An interesting choice is the the well-known self-financing strategy where the issuer will not admit any loss until the contract reach its maturity. Hence, the issuer rebalances the hedging portfolio using the accumulated value less the benefit paid at this time, that is
%\[
%W_{i_{t-1}} ( x_{i_{t-1}}, \xi^i_t , z_{i_{t-1}}) -P_{i_t} = C_{i_t},
%\]     
%for $t=1,...,T-1$.  

\subsection{Barrier on Future Risk}

Even though Algorithm \ref{algo5} minimizes a weighed average of risk measures, we noticed during backtesting that the risk measure could take for some nodes unacceptable high values. In the following algorithm we try to remedy this drawback by constraining the future CVaR.
\begin{algo}\label{algo3}
For all $t=T, T-1,  \cdots, 1$ and all $i_{t-1}\in I_{t-1}$,  
\begin{align}
V_{i_{t-1}}(z_{i_{t-1}})= \min_{ x_{i_{t-1}},z_{i_{t}} }  \textnormal{CVaR}_c( L_{i_{t-1}} (x_{i_{t-1}}, ~\xi^i_t,  G_{i_t}, z_{i_{t-1}} ))\label{rmo}
\end{align}
Under the constraints
\begin{align}
H_{i_t}( z_{i_{t-1}}, z_{i_t},x_{i_{t-1}},\xi^i_t, G_{i_t}  )&= 0, \label{rmc1}\\
V_{i_{t}}(z_{i_{t}})&\leq \gamma_0,~~ \forall   i_t  \in I_t|i_{t-1} , \label{rmc2}
\end{align}
where $V_{i_T}(z_{i_T})=0$ for all $i_T$.
\end{algo}
The parameter $\gamma_0$ controls the future risk allowed by the issuer. There are alternative ways to constraint future risk measure including, for example, the expected value. We prefer constraint (\ref{rmc2}) since it better controls the path wise risk.% Note that the constraint (\ref{rmc2}) does not exist for $t=T$.

Here again, at least one of the state variables has to be the initial capital needed for financing the investment strategy.  If we let the first components of the state vector be this capital then the first state equation has to be $z_{i_{t-1}} =F_{i_{t-1}} (x_{i_{t-1}} )$. 

There is a fundamental difference between Algorithm \ref{algo3} and standard stochastic dynamic programs since the cost-to-go function is not anymore modelled in the objective function, but used instead in the constraint.

%The algorithm feasibility between $t-1$ and $t$ is conditional on the issuer decision on the continuation value. Therefore the continuation values are determined based on the level of risk that the issuer is ready to accept in the following period. This risk level threshold is set to $\gamma_0$. In term of constraint, we introduce (\ref{rmc2}), where the state variable $z_{i_{t}}$ is the continuation value.

In the case of the hedging portfolio consists of three assets, according to \cite{Rockafellar}, the objective function of Algorithm \ref{algo3} becomes
 \begin{align}
V_{i_{t-1}}(z_{i_{t-1}})= \min_{x_{i_t},\pi_c^+,\pi_c^-, u_{i_{t}} ,z_{i_t }}   \pi_c^+ - \pi_c^- + \sum_{i_t} \dfrac{u_{i_t} p{_{i_t|i_{t-1}}}  }{1-c},\label{rmob}\end{align}
with following constraints
\begin{align}
-G_{i_t} (P_{i_t} , z_{i_t})  + a_{i_{t-1}}\frac{S_{i_t}}{S_{i_{t-1}}}+b_{i_{t-1}}e^{r}+c_{i_{t-1}}\frac{O_{i_t}(T)}{O_{i_{t-1}}(T)}+ \pi_c^+ - \pi_c^- + u_{i_t} &\geq 0,~~  \forall   i_t  \in I_t|i_{t-1} ,\nonumber \\
u_{i_t} & \geq 0,~~ \forall   i_t  \in I_t|i_{t-1} ,\nonumber \\
\pi_c^-,\pi_c^+ &\geq 0,\nonumber\\
a_{i_{t-1}}+b_{i_{t-1}}+c_{i_{t-1}}&=z_{i_{t-1}},\nonumber\\
 V_{i_{t}}(z_{i_{t}}) &\leq \gamma_0 ,~~ \forall   i_t  \in I_t|i_{t-1}.\nonumber
\end{align}
Similarly to the optimization problem with the objective function (\ref{spab}), $V_{i_t}$ is replaced by its supporting hyperplane and the previous optimization problem is also a right hand side parametric linear program with parameter $z_{i_t}$. Hence, this algorithm shares the same properties as well as complexity and use the same optimization technics as Algorithm \ref{algo5}.  

\subsubsection{Coherent Dynamic Risk Measure}
%Dans ce modele la continuite  zit 
%The previous algorithm could benefit from linking the different periods together to be more consistent and meaningful. 

The following algorithm uses the concept of coherent dynamic risk measures to link the different periods. \cite{Riedel04} shows that coherent dynamic risk measures satisfy a discounted recursion, that is
\begin{algo}\label{algo4}
For all $t=T, T-1,  \cdots, 1$ and all $i_{t-1}\in I_{t-1}$,  
\begin{align}
V_{i_{t-1}}(z_{i_{t-1}})= \min_{ x_{i_{t-1}},z_{i_{t}} }   e^{-r} \textnormal{CVaR}_c( L_{i_{t-1}} (x_{i_{t-1}}, ~\xi^i_t,  G_{i_t}, z_{i_{t-1}} )+V_{i_t}(z_{i_t}))\label{cdo}
\end{align}
Under the constraint
\begin{align}
H_{i_t}( z_{i_{t-1}},z_{i_t},x_{i_{t-1}},\xi^i_t , G_{i_t} )=0 \label{cdr}
\end{align}
where $V_{i_T}(z_{i_T})=0$ for all $i_T$.
\end{algo}%where $\lambda$ is the proportion of the future CVaR that the objective function considers. 
In the objective function the risk measure takes into account the loss random functions as well as the subsequent near future risk measure. This allows to keep some path wise risk control.

Again at least one of the state variables has to be the initial capital needed for financing the investment strategy. In this case the state variable $z_{i_t}$ is the continuation value $C_{i_{t}}(z_{i_{t}})=z_{i_t}$. 
 
 In the case we consider the hedging strategy involving three assets according to \cite{Rockafellar} the Algorithm \ref{algo4} becomes
\begin{align}
V_{i_{t-1}}(z_{i_{t-1}})=e^{-r} \min_{x_{i_t},\pi_c^+,\pi_c^-, u_{i_{t}},\theta^+_{i_{t}},\theta^-_{i_{t}}  }   \pi_c^+ - \pi_c^- + \sum_{i_t} \dfrac{u_{i_t} p{_{i_t|i_{t-1}}}  }{1-c} ,\label{cdob}
\end{align}
under the constraints
\begin{align}
-L_{i_{t-1}} (x_{i_{t-1}}, ~\xi^i_t,  G_{i_t}, z_{i_{t-1}}) + \pi_c^+ - \pi_c^- + u_{i_t} -\theta^+_{i_t}+\theta^-_{i_t}&\geq 0,~~  \forall   i_t  \in I_t|i_{t-1} ,\label{cdob1}\\
u_{i_t} & \geq 0,~~ \forall   i_t  \in I_t|i_{t-1} ,\label{cdob2} \\
\theta^+_{i_t}-\theta^-_{i_t}-V_{i_t}(z_{i_t})& \geq 0,~~ \forall   i_t  \in I_t|i_{t-1} \label{cdob3}\\
a_{i_{t-1}}+b_{i_{t-1}}+c_{i_{t-1}}&=z_{i_{t-1}},\label{cdob4}
\end{align}
where $L_{i_{t-1}}$ is given by (\ref{wtm1}) and (\ref{ltm1}). $\theta^+$ and $\theta^-$ are a split of a variable that could be positive or negative for the need of the linear programming solvers. The cost-to-go function $V_{i_t}$ is modelled through its epigraphic definition. If $V_{i_t}$ is a convex piecewise linear function defined by its supporting hyperplane, the constraint $\theta^+_{i_t}-\theta^-_{i_t}-V_{i_t}(z_{i_t}) \geq 0$ is equivalent to several linear constraints where each of $V_{i_t}(z_{i_t})$ supporting hyperplane is smaller or equal to $\theta^+_{i_t}-\theta^-_{i_t}$. 

Similarly to Algorithms \ref{algo5} and \ref{algo3}, the sate variable $z_{i_t}$ is the amount required to continue the issuer's operations beyond $i_t$. Properties and complexity are also similar to the previous optimization problems and analogous technics can be used.% In this case, the continuation value $C_{i_t}$ needs to be used as the state variable since the optimization problem minimizes the risk measure at node $i_t$.

\begin{proposition} 
The previous optimization problem is equivalent to Algorithm \ref{algo4}.  
\end{proposition} 

 \begin{proof}
What we really have to prove is that for every binding constraint of ($\ref{cdob1}$) the associated constraint of ($\ref{cdob3}$) is also binding (i.e. $\theta^{*+}_{i_t}-\theta^{-*}_{i_t}=V_{i_t}(z^*_{i_t})$, where  $\theta^{*+}_{i_t}, \theta^{*-}_{i_t}$, and $z^*_{i_t}$ are the optimal values of $\theta^{+}_{i_t}, \theta^{-}_{i_t}$, and $z_{i_t}$, respectively). We will proceed by contradiction and assume $\theta^{*+}_{i_t}-\theta^{-*}_{i_t}>V_{i_t}(z^*_{i_t})$.  
Let $\lambda_{i_{t}}^* <0$ the Langrange multiplier relative to $i_{t}$ in the constraint ($\ref{cdob1}$). Since we are dealing with convex optimization, the optimization problem (\ref{cdob})-(\ref{cdob4}) is equivalent to 
\begin{align}
 \min_{x_{i_t},\pi_c^+,\pi_c^-, u_{i_{t}},\theta^+_{i_{t}},\theta^-_{i_{t}}  }&&   \pi_c^+ - \pi_c^- + \sum_{i_t} \dfrac{u_{i_t} p{_{i_t|i_{t-1}}}  }{1-c}+\lambda_{i_{t}}^*(-L_{i_{t-1}} (x_{i_{t-1}}, ~\xi^i_t,  G_{i_t}, z_{i_{t-1}}) \nonumber \\
 &&+ \pi_c^+ - \pi_c^- + u_{i_t} -\theta^+_{i_t}+\theta^-_{i_t})\label{cdobp1}
\end{align}
with the same constraints ($\ref{cdob1}$)-($\ref{cdob4}$) except for the specific $i_t$ in ($\ref{cdob1}$). We know from linear optimization that either $\theta^+_{i_{t}}$ or $\theta^-_{i_{t}}$ is active i.e. that at optimality  only one of the two is positive. If $\theta_{i_t}^{*+}$ (rest. $\theta_{i_t}^{*-}$) is strictly positive, the objective function (\ref{cdobp1}) decreases when $\theta_{i_t}^{*+}$ decreases (resp. increases), since $\lambda_{i_{t}}^* <0$. The decrease is possible since by assumption $\theta^{*+}_{i_t}$(resp. $-\theta^{*-}_{i_t}$)$>V_{i_t}(z^*_{i_t})$. In both cases we have a contradiction. 
 \end{proof}

% objective function (\ref{cdo}) becomes (\ref{rmo1}) and in addition to the constraint (\ref{cdr}), the following constraints need to be added
                                    
\section{Numerical Examples}
\label{sec:numex}

Our numerical examples consist of Guaranteed Investment Certificates (GIC) and Equity-Index Annuity (EIA), where the index is governed by a Binomial recombining tree. It is assumed that the period is fixed to one month and each period is divided into $N$ subperiods. Hence, the index makes $N$ movements per periods and the issuers observe and re-balance their hedging portfolios each month according to the risk-control strategies. Hence, the total number of periods $T$ of the planning horizon is expressed in terms of months. 

For simplification purposes, the index will be governed by the \cite{crr79} model where $S_{0_0}=1$. In this recombining tree, the index $S_{i_t}$ has two possible outcomes in each subperiod: it either increases to $uS_{i_t}$ or decreases to $dS_{i_t}$. The increasing and decreasing factors $u$ and $d$ are assumed such that the index process converges to the lognormal distribution underlying a continuous geometric Brownian motion with annual drift $\mu$ and volatility $\sigma$. In other words, $u=e^{\sigma (12N)^{-0.5}}$ and $d=u^{-1}$, where the volatility is assumed to be equal to 20\%. It also requires that the probability of an up movement is defined by
\begin{equation}
\Pr[S_{i_{t+n/N}}=uS_{i_{t+(n-1)/N}}|S_{i_{t+(n-1)/N}}] =\frac{e^{\mu/12}-d}{u-d}=\pi, \label{pi}
\end{equation}
for $t=0, 1, ...$ and $n=1,....,N$.% The special case where $n=1$ gives the value of the stock at the end of the first subperiod and when $n=N$ is the value of random process $\xi$ at the next period. 

At period $t$, the index process is distributed according to a binomial distribution given $S_{i_{t-1}}$. This conditional process has $N+1$ possible outcomes at $t$ given by $S_{i_{t-1}}u^{j}d^{N-j}$, $j=0,...,N$ with corresponding conditional probabilities
\begin{eqnarray}
\Pr[S_{i_{t}}=S_{i_{t-1}}u^{j}d^{N-j}|S_{i_{t-1}}]= \left(
\begin{array}{c}
  N \\
  j \\
\end{array}%
\right) \pi ^{j} [1-\pi]^{N-j}, ~j=0,1,..., N.
\label{bin1}
\end{eqnarray}

Hence the algorithms face a stochastic process based on a recombining tree with the cardinality of $I_t|i_{t-1}$ equal to $N+1$. The nodes in the tree increase linearly by $N$ in each period. In the case of GICs the stochastic process $\xi$ is nothing but $S$. In the case of EIAs, the stochastic process has two dimensions: $S$ and the policyholder cohort.

For the sake of simplicity, we assume the usual frictionless financial market: no tax, no transaction costs, etc.  

\subsection{Hedging Error Mismatches}

The efficiency of the different strategies is assessed using the discounted value of the path wise error incurred by the strategy. The error is given by the temporary loss function. Hence, the present value hedging error is defined by
\begin{align}
M=\sum_{t=1}^{T} e^{-r t/12} L _{i_{t-1}} \label{m1},
\end{align}
where $r$ is the annual rate.

For comparative purposes, we add the initial value of the hedging portfolio $F_{i_0}(x_{i_0})$ and CVaR$_{95\%}$ of the discounted mismatch $M$ and subtract the initial investment from the holder, that is $CR=F_{i_0}(x_{i_0})+$CVaR$_{95\%}(M)-1$, since the initial investment is assumed to be 1. In other words, the issuer holding the amount $CR$ at the beginning possesses the expected value of the worst 5\% losses. 

This measure $CR$ is used as a guideline to fine tune the risk measure parameters (e.g. $c$, $\gamma_0$, etc.) in all algorithms. Different sets of these parameters are tested and the ``optimal" values are compared to asses the quality of each hedging strategy.

\subsection{Guaranteed Investment Certificates}

Guaranteed Investment Certificates are secure investments issued by banks that offer a wide range of options. They are secure since the investor's principal is guaranteed by issuers. GIC can be split in two main categories: cashable and non-cashable. The non-cashable GICs can provide either a fixed return or a return linked with the financial market. In this case the GIC can also guarantee a minimum return. The issuers reduce the cost of such contracts by imposing an upper bound on the market return called a cap. The GIC payoff is given by 
\begin{align}
P_{i_T}=\max \left[\min \left[ 1+R_{i_T},(1+\zeta)^{T/12}\right],(1+g)^{T/12} \right],\label{gic}
\end{align}
where $\zeta$ is the annual cap rate, $g$ minimum annual guaranteed rate, and the return $R$ is defined by
\begin{equation}
  R_{i_T}=\frac{S_{i_T}}{S_0}-1. \label{RT}
\end{equation}%=\prod_{t=1}^{T} \xi^i_t -1
The minimum in (\ref{gic}) imposes a barrier on the possible returns and the maximum guarantees a minimum return on the investment,

For GICs, the continuation value $G_{i_t}(P_{i_t},C_{i_t})=C_{i_t}$, for $t=1,....,T-1$ since we do not have any intermediate payments. In addition, the random process $\xi^i_t$ is simply $S_{i_t}$ for $t=0,...,T$.

In absence of indication to the contrary, we assume the following for our analysis: 1-year maturity contracts ($T=12$), $g=0\%$, $\zeta=6\%$, $r=3\%$, $\mu=8\%$, $\sigma=20\%$, and $N=6$.

In addition to a risk-free asset and the underlying stock, it is assumed that the issuer can invest, every month, in a 1-month European call option at the money. Without loss of generality, we assume that at time $t$ the issuer can invest in a European call option, which is at the money. The payoff of this option is given by $(S_{i_{t+1}}-S_{i_t})^+$, and the price $O_{i_t}(t+1)$ is obtained using \cite{BS73}. In this paper, we omit to consider financial options with longer maturities to avoid more complex modelling and algorithms. We implemented maturities that cover several optimization periods and observed similar conclusions about the impact of adding the options to the set of assets.

%Let  denote the time-$t$ price of the financial option with maturity $t+1$.

For each hedging strategy, the optimal replicating portfolios are obtained using the backward approach presented in Section \ref{HS}. Once the optimal hedging portfolios are known, the distribution of hedging errors is estimated. Since errors calculations imply decisions at each node, we have to deal with the associated unfolded tree. In this case, the number of nodes increases exponentially with the number of periods. For example, the unfolded tree has $(N+1)^{12}$ different final nodes at the last period. Consequently we proceed by sampling different index path scenarios. We use 50,000 simulations to approximate the distribution of discounted errors $M$ defined in (\ref{m1}), which allows estimating $CR$.

\subsubsection{Risk Measure as Constraint}

We first investigate the stability of our algorithm. We obtained the initial value of the hedging portfolios using Algorithm \ref{algo2} for different periods ($T=2, 4, 6, 8, 12, 24$) and nodes ($N=2, 4, 6, 8, 12, 24$). For  comparison purposes, the risk measure threshold $c$ is set to 60\% and parameter $\gamma_0$ is set to 0. The cost is increasing when both parameters are increasing.  Increasing the number of nodes increases the risk by enlarging the range of the index distribution. Increasing the number of period increases the number of times that the issuer is monitoring his portfolio, which is costlier. The algorithm is fairly stable when both the number of periods and nodes are larger than 6. Hence, increasing the number of periods and nodes above 6 has a minimal impact on the initial value.
 
 \begin{center}
\begin{table}[htbp]
  \centering
  \begin{tabular}{c||c  c  c  c c c   }%\cline{2-13} 
  \hline
 $T$ $\backslash$  $N$ &2&4&6&8&12&24\\
    \hline

2	&	0.9948	&	1.0045	&	1.0081	&	1.0108	&	1.0124	&	1.0151	\\
4	&	1.0023	&	1.0109	&	1.0113	&	1.0128	&	1.0139	&	1.0122	\\
6	&	1.0063	&	1.0135	&	1.0127	&	1.0134	&	1.0115	&	1.0126	\\
8	&	1.0089	&	1.0150	&	1.0132	&	1.0113	&	1.0111	&	1.0134	\\
12	&	1.0122	&	1.0164	&	1.0108	&	1.0103	&	1.0116	&	1.0127	\\
24	&	1.0165	&	1.0125	&	1.0112	&	1.0114	&	1.0113	&	1.0127	\\ 
\hline

\end{tabular}%
    \caption{Initial portfolio values $F_{i_0}$ for different numbers of periods $T$ and nodes $N$.}
  \label{tabstab}%
\end{table}%\end{center}
\end{center}

We now set the risk measure retention level $c$ and the threshold $\gamma_0$ for the CVaR as a constraint (Algorithm \ref{algo2}). In Figure \ref{rgic1_c}, the $CR$ are evaluated for different risk measure levels ($c=5\%,6\%,\cdots, 95\%$) and different $\gamma_0$ (0, 1\%, 2\%) based on Algorithm \ref{algo3}. The CVaR$_{59\%}$ and $\gamma_0=0$ provide the best results for the mismatches since the $CR$ reaches a minimal value of 1.14\%. The retention levels $c$ between 45\% and 65\% perform well. The threshold $\gamma_0$ does not seem to influence the performance of the hedging portfolio when the evaluation is based on $CR$. For the rest of the numerical examples, the parameter $\gamma_0$ is set to 0. 

\begin{figure}[h!]
 \includegraphics[scale=0.75] {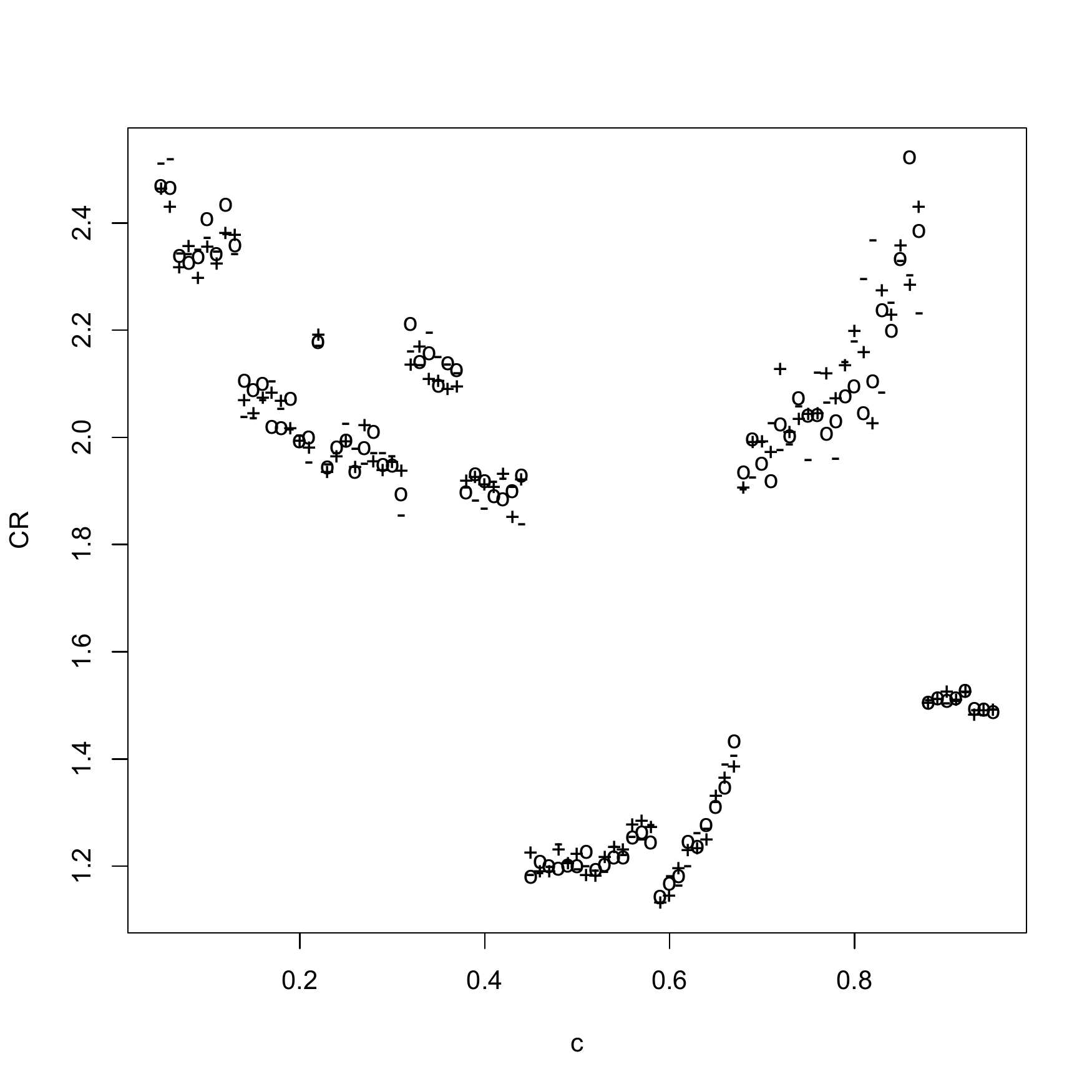}
 \caption{$CR$ for different retention level $c$. The symbols $+$, $-$, and o refer to $\gamma_0$ equal to $0$, $1\%$, and $2\%$, respectively. }
 \label{rgic1_c}
\end{figure}
%matplot(5:95/100,t(rgic1_c4),xlab="c",ylab="CR",main=c(),pch=c("+","-","o"),col=1) 
%legend("top",legend=c("+ 0%, - 1%, o 2%"),bty="n")

Figure \ref{gic1} shows the distribution of the hedging errors where the retention level is fixed to $c=59\%$. This algorithm provides a left-skewed distribution for the mismatches, where the issuer possessing 0.87\% of capital will face losses with a probability of 5\%. The initial value to start the hedging portfolio $F_{i_0}$ is 1.01. The expected gain from the transaction is 0.33\% with a fairly small standard deviation (1.29\%).

\begin{figure}[h!]
 \includegraphics[scale=0.75] {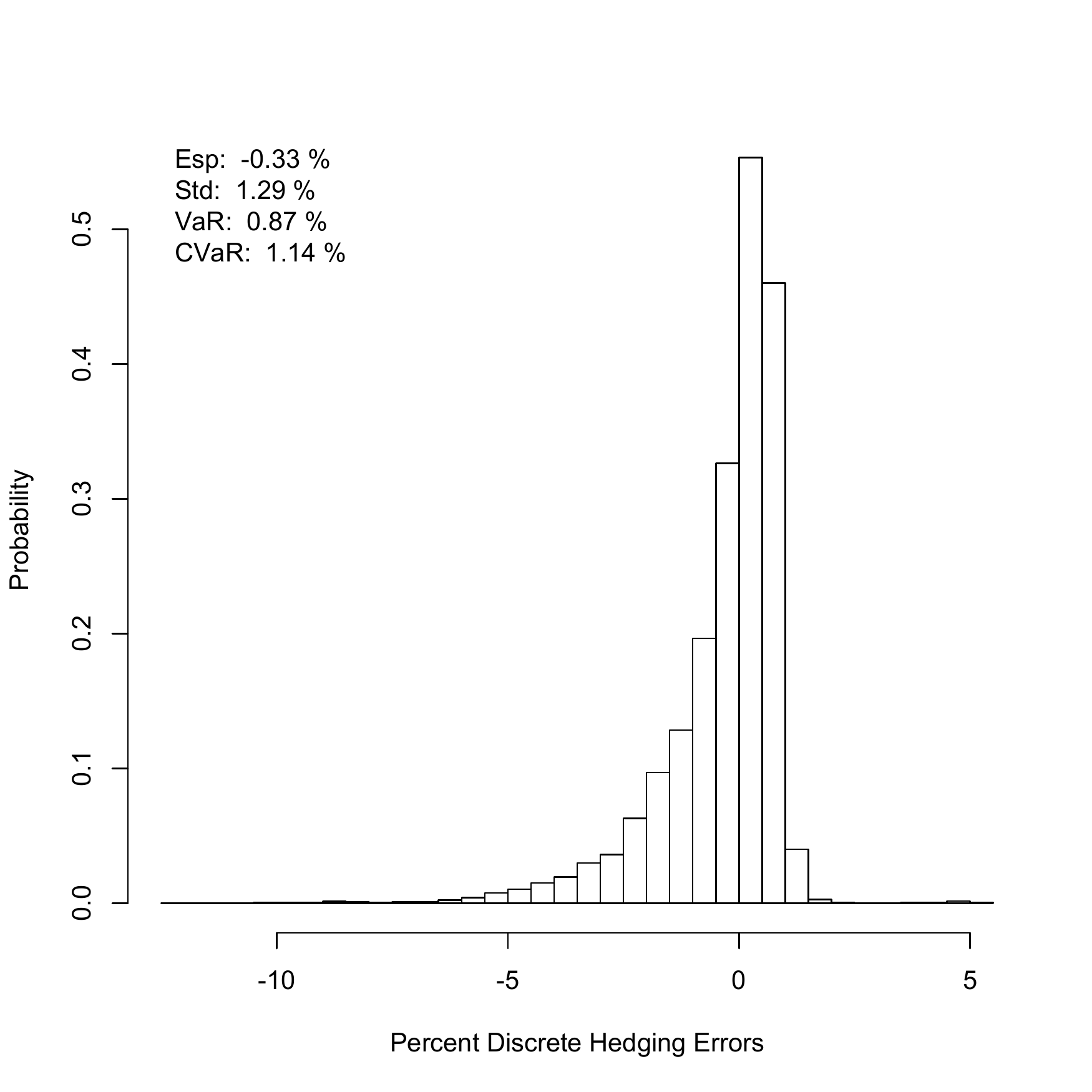}
 \caption{Distribution of the hedging errors for $c=59\%$.}
 \label{gic1}
\end{figure}
%gic1(cc=0.59)

%n2
Table \ref{tabn1} gives the hedging statistics for different numbers of trading dates ($2, 4, 6, 8,12, 24, 52$).  All parameters in the tables are reported in percent. As explained in Gaillardetz and Moghtadai (2017), changing parameters requires retuning the risk measure retention level.  The table provides optimized results for the risk measure threshold. The mismatches are generated for different $c$ $(5\%, 10\%, ,\cdots, 95\%)$ and the lowest $CR$ obtained from the simulations is presented as well as the risk level that provided this result. Table \ref{tabn1} shows that increasing the number of trading dates impact the distribution of the mismatches. The first observation is that the amount $CR$ is decreasing as the number of re-balancing is increasing since the replicating portfolio controls the risk more frequently. This is not observable for small numbers of trading dates. It is clear that the optimal hedging strategy is obtained using a risk measure threshold $c$ equal to 50\% for large numbers of trading dates. The decrease is steady since high levels are used for small numbers of trading dates.
 \begin{center}
\begin{table}[htbp]
  \centering

  \begin{tabular}{c||c  c  c  c c c c   }%\cline{2-13} 
  \hline
 $T$ &2&4&6&8&12&24&52 \\
    \hline
 
%     \multicolumn{1}{c}{}  & \multicolumn{7}{c}{number of trading dates}\\
%    \hline
$c$	&	95	&	95	&	95	&	65	&	60	&	50	&	50	\\ 
$CR$	&	1.05	&	1.58	&	1.80	&	1.80	&	1.14	&	1.02	&	0.84	\\ \hline
\end{tabular}%
    \caption{Results for different numbers of trading dates or periods $T$.}
  \label{tabn1}%
\end{table}%\end{center}
\end{center}

The number of nodes is also affecting the algorithm performance. In Table \ref{tabn2}, we set the number of nodes $N$ to 2, 4, 6, 8, 12, 24, and 52. The measurement $CR$ is increasing with the number of nodes since increasing $N$ increases the tails of the index distribution. It is not clear which risk measure threshold that should be used for large $N$.  

\begin{center}
\begin{table}[htbp]
  \centering

  \begin{tabular}{c|c  c  c  c c c c   }%\cline{2-13} 
  \hline
  $N$&2&4&6&8&12&24&52 \\
    \hline
 
$c$	&	40	&	90	&	60	&	40	&	60	&	70	&	55	\\
$CR$	&	0.96	&	1.80	&	1.14	&	1.23	&	1.32	&	1.38	&	1.44	\\ \hline
\end{tabular}%
    \caption{The amount $CR$ for different numbers of nodes $N$.}
  \label{tabn2}%
\end{table}%\end{center}
\end{center}

%g,cap,T
Table \ref{tgg} shows the value $CR$ for different GIC contracts. The algorithm is performed under different maturities ($T=24, 36$), cap rates ($\zeta=5\%, 7\%$) as well as guaranteed rates ($g= 1\%, 2\%$).  Extending the contract maturity does require more capital since the guarantee is expended over several years. However, the amount $CR$ is still reasonable going from 1.14\% for one year to 1.87\% for 3 years (first block). The cap rate is the parameter that influences the most $CR$. It is linked with the tail of the possible return, which is limited by the cap rates. The last block shows that the algorithm is fairly robust when increasing the guaranteed return, which influences the floor of the possible return.     

\begin{center}
\begin{table}[htbp]
  \centering

  \begin{tabular}{c|c  c || c|  c c ||c|c c   }%\cline{2-13} 
  \hline
  $T$ &$c$& $CR$&$\zeta$ &$c$& $CR$&$g$ &$c$& $CR$ \\
    \hline
 
24	&	45	&	1.71	&	5\%	&	60	&	0.73	&	1\%	&	60	&	1.39	\\
36	&	45	&	1.87	&	7\%	&	45	&	1.67	&	2\%	&	45	&	1.62	\\ \hline
\end{tabular}%
    \caption{Results for different GIC contracts.}
  \label{tgg}%
\end{table}%\end{center}
\end{center}

%mu, r, and vol
A sensitivity analysis is performed in Table \ref{tpar} for diverse financial parameters. The volatility is set to 15\% and 25\%. The short rate and expected return $\mu$ vary between 2\% and 4\% and between 7\% and 9\%, respectively. The values $CR$ are fairly stable to volatility variations (see second block) and the algorithm is robust to the change of the expected return (see first block). Even though the underlying contract has a 1-year maturity, the changes in the short rate have major impacts on the portfolio initial values and the performance of the strategies as seen in the third block.  

%option
Table \ref{tpar} stresses the impact of introducing the European option in our hedging strategy. The last block shows the amount $CR$ for the hedging strategy without option. For a one year contract, the issuer reduces the $CR$ from $1.86$ to $1.14$ with the options. The hedging strategy outperforms the approach proposed by \cite{Gaillardetz17}. 

\begin{center}
\begin{table}[htbp]
  \centering

  \begin{tabular}{c|c  c || c|  c c ||c|c c ||c|c c   }%\cline{2-13} 
  \hline
  $\mu$ &$c$& $CR$&$\sigma$ &$c$& $CR$&$r$ &$c$& $CR$ &Option &$c$& $CR$\\
    \hline
 
7\%	&	60	&	1.16	&	15\%	&	45	&	1.13	&	2\%	&	60	&	2.04	&	with	&	59	&	1.14	\\
9\%	&	50	&	1.17	&	25\%	&	60	&	1.46	&	4\%	&	60	&	0.31	&	without	&	60	&	1.86	\\ \hline
\end{tabular}%
    \caption{Sensitivity analysis.}
  \label{tpar}%
\end{table}%\end{center}
\end{center}

%ei
The issuer could be interested to diminish the left part of the hedging errors tails from Figure \ref{gic1}. Using the path wise hedging error as a state variable, (\ref{eiT}) and (\ref{eit}), the issuer could restrain the aggregate error by using the parameter $\gamma_3$. Figure \ref{gice1} shows the present value hedging error under the same setting as Figure (\ref{gic1}) except that we limit each error path to 3\%. This strategy possesses similar statistics related to the hedging errors (expected gain, standard deviation, VaR, and $CR$) than the other standard strategy. The difference between the strategies is the range of the right tail; the model removes the losses above 3\% by replacing high losses with more frequent smaller losses.  

\begin{figure}[h!]
 \includegraphics[scale=0.75] {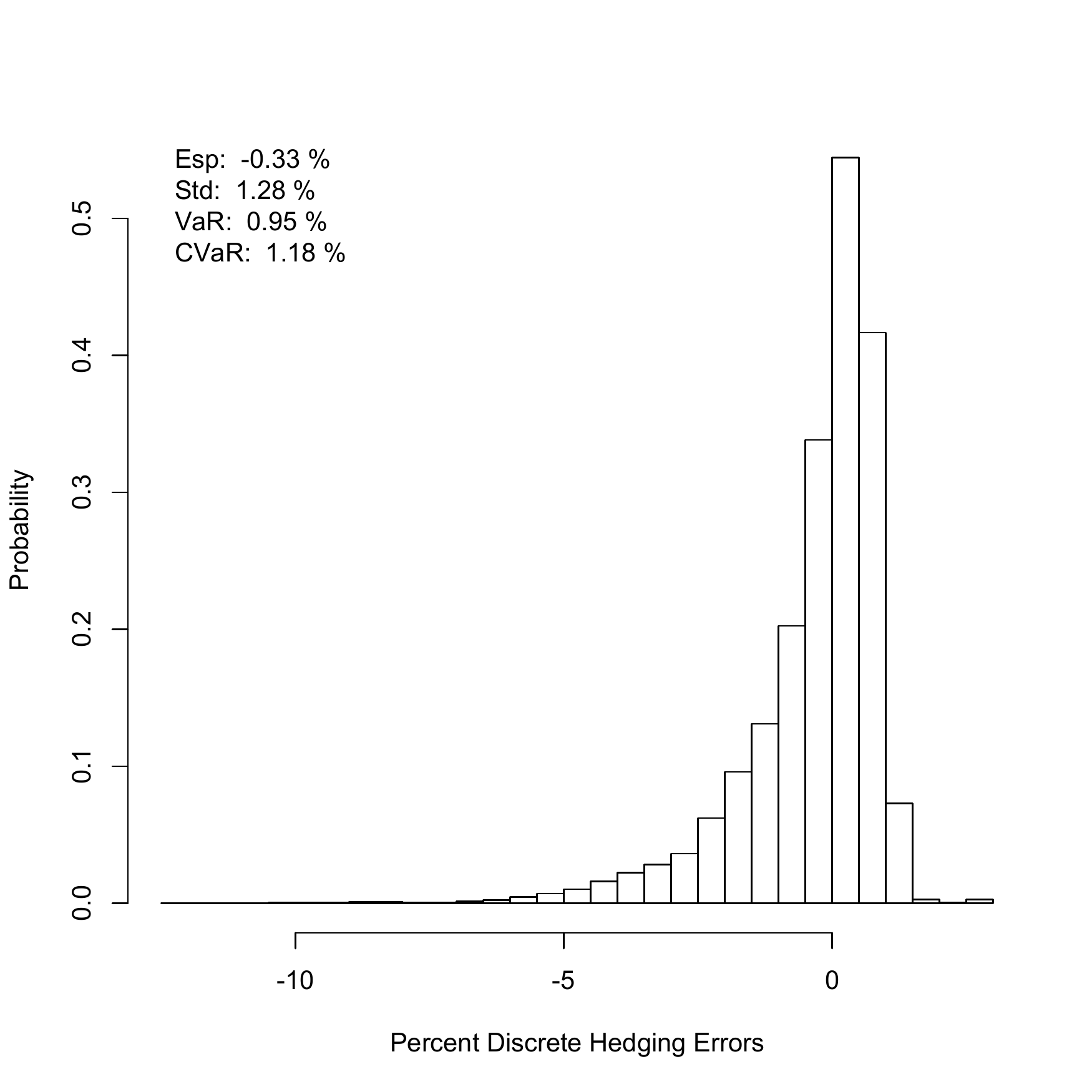}
 \caption{Distribution of the hedging errors for $c=59\%$ and $\gamma_3=3\%$.}
 \label{gice1}
\end{figure}

\subsubsection{Risk Measure as Objective}

The Algorithms \ref{algo5}, \ref{algo3}, and \ref{algo4} are also implemented for the default GIC contract. Since the Algorithms \ref{algo5} and \ref{algo4} do not impose a lid on risk measure, we notice that for some nodes the optimization results yield to unacceptable values for risk measure. In order to make these algorithms more efficient we need additional local control on the mismatches  (\ref{cnormp})-(\ref{cnorm2}) or path wise constraints (\ref{eiT})-(\ref{eit}); these constrains reduce the highest losses in each transition. For example, the Algorithm \ref{algo4} is able to perform well when the local CVaR of the hedging errors is also limited to $0$. Figure \ref{gic3} gives the distribution of the hedging errors for Algorithm \ref{algo4}, where the risk measure level is optimized and set to 59\%. The hedging statistics are similar to the ones obtained from Algorithm \ref{algo2} with an initial value of 1.03 for the replicating portfolio.  

Algorithm \ref{algo3} performs well but slightly less good than Algorithm \ref{algo4} for this product. We omit to report the results for the sake of compactness.

%Based on our numerical analyses the Algorithms \ref{algo5}, and \ref{algo3} require additional local or global controls. Constraints that control the norm of the errors  (\ref{cnormp}) or the aggregate errors (\ref{eiT}) could be added to reduce the risk. The Algorithm \ref{algo4} quite well 

%The algorithms require additional control when the risk measure is used in the objective function.

%  In addition to their respective constraints, the local CVaR of the loss is also limited to $\gamma_0$. The Algorithm \ref{algo3} and \ref{algo5}  were not able to perform well and did not provide competitive hedging statistics compared to Algorithm \ref{algo2} for this GIC contract. %Double check numbers
\begin{figure}[h!]
 \includegraphics[scale=0.75] {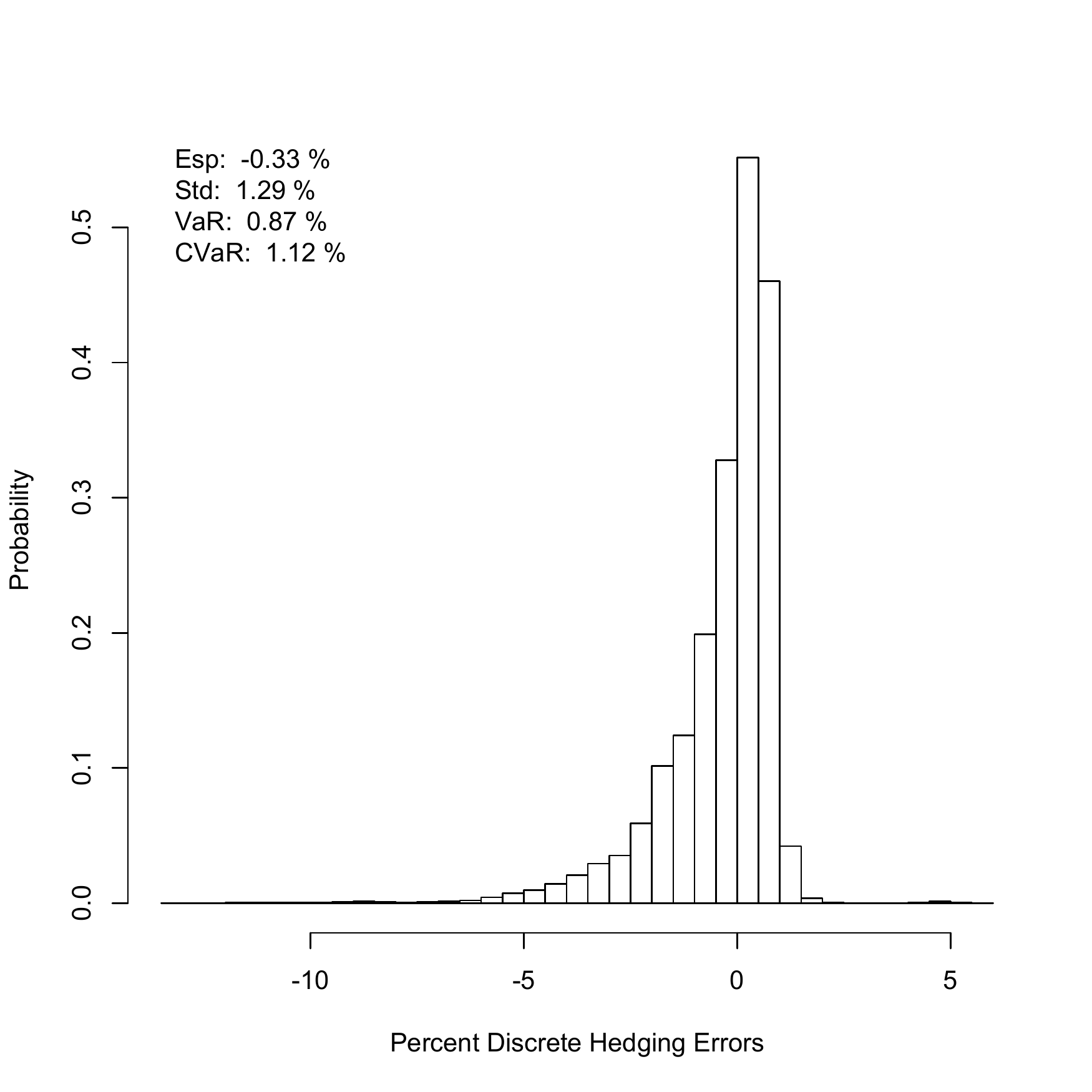}
 \caption{Distribution of the hedging errors for $c=59\%$ under the coherent risk measure approach.}
 \label{gic3}
\end{figure}

% and did not provide competitive hedging statistics compared to Algorithm \ref{algo2} for this GIC contract
%These algorithms provide the issuer some flexibilities in the choice of the hedging portfolio. We opt for a self-financing strategy, where we start with one unit and use the accumulated value as a starting value. Some adjustments are sometimes necessary to remain in the feasible region.

\subsection{Equity-Indexed Annuities}

Equity-indexed annuities are similar to GICs since they are linked to the performance of equity index while offering protections. However, the investors may usually opt for diverse financial guarantees, such as \emph{Guaranteed Minimum Death
Benefits} (GMDB) as well as \emph{Guaranteed Minimum Living Benefits} (GMLB).  %The simplest form is the \emph{Return of Premium Death Benefit} where the maximum of the current account value at the time of death and the single premium is paid.
%The risk profile of the investor is set by the specific selection of mutual funds.

The guaranteed minimum death benefits (GMDB) pays a guaranteed amount in case the insured dies during the deferred period. The guaranteed minimum living benefits (GMLB) are separated into three types: \emph{Guaranteed Minimum Accumulation Benefits} (GMAB), \emph{Guaranteed Minimum Income Benefits} (GMIB) and more recently \emph{Guaranteed Minimum Withdrawal Benefit} (GMWB). GMAB is the simplest form of these benefits, where the insured is entitled to the single premium or a roll-up benefit base at maturity. The GMIB offers the choice to obtain the account value, annuitize the account value or annuitize some guaranteed amount at specified rates. The GMWB gives the possibility to withdraw a certain amount in small portions annually. See the monograph by \cite{Hardy03} for comprehensive discussions on these guarantees. The focus will be on GMDB and GMAB guarantees, where the hedging portfolio can be obtained using the proposed algorithms. 

To illustrate an equity-linked product valuation, we consider one of the simplest
design of EIAs, known as the point-to-point with term-end design where the index growth is based on the growth between two time points over the entire term of the annuity. This design has embedded GMDB and GMAB insurances with the payoff at time $k$ represented by
\begin{align}
P_{i_t}=\max \left[\min \left[ 1+\alpha R_{i_t},(1+\zeta)^{t/12}\right],\beta(1+g)^{t/12} \right],\label{ptp}
\end{align}
for $t=1,\cdots,T$, where $\alpha$ is the participation in the index and the return is given by (\ref{RT}). EIAs provide participation in the index return at the level $\alpha$ as well as protection against the loss from a down market $\beta(1+g)^k$. The main difference between this EIA and the GIC is that the EIA contract terminates when the investor died and it is assumed that the insurance company pay the benefit at the end of the period. In this case, the continuation value 
\[
G_{i_t}(P_{i_t},C_{i_t})=P_{i_t}1_{\{T_x\in[t,t+1)\}} +C_{i_t} 1_{\{T_x\geq t+1\}},
\]
for $t=1,...,T-1$ and $C_{i_T} =P_{i_T}$, where $1_{\{ .\} }$ is an indicator function and $T_x$ is the future lifetime (in terms of periods) of the policyholder ($x$).

The binomial structure remains the same for the index. The stochastic process has 2 dimensions: the index process and the cohort evolution. It is assumed that both processes are independent and that for sake of simplicity our insurance portfolio consists of one policyholder. Hence, the cohort process has two states: death or alive with probability $q$ and $1-q$, respectively.

The cardinality of the transitions set $I_t|i_{t-1}$ is $(N+1)*2$. The conditional possible values for $\xi^t_i $ given node $i_{t-1}$ are
\begin{eqnarray}
\left\{\begin{array}{l l }
 (S_{i_{t-1}} u^{j}~ d^{N-j},1_{\{T_x\in[t-1,t)\}}) & j=0,1,\cdots, N\\
 (S_{i_{t-1}} u^{j-N-1}~ d^{2N+1-j},1_{\{T_x\geq t\}} )& j=N+1,N+2,\cdots, 2N+1.
\end{array}\right.
\end{eqnarray}
The first $N+1$ transitions are related to the death of the policyholder in the given period and the following $N+1$ nodes are linked to the survival of $(x)$. The corresponding conditional probabilities $p_{i_t|i_{t-1}}$ are successively
\begin{eqnarray}
\left\{\begin{array}{l l }
 \left(
\begin{array}{c}
  N \\
j\\
\end{array}%
\right) \pi ^{j} [1-\pi]^{N-j}~ q_{t-1}, & j=0,1,\cdots, N.\\

 \left(
\begin{array}{c}
  2N+1 \\
 j-N-1\\
\end{array}%
\right) \pi ^{j-N-1} [1-\pi]^{2N+1-j} (1-q_{t-1}), & j=N+1,N+2,\cdots, 2N+1,
\end{array}\right.
\end{eqnarray}
where $q_{t-1}=\Pr[t-1\leq T_x <t ~|T_x\geq t-1]$ .

In absence of indication to the contrary, we assume the following for our analysis: 5-year contracts $(T=60)$, $\alpha$=50\%, $\beta=100\%$, $g=0\%$, and $\zeta=\infty$. In the case of EIAs, it is also assumed that the policyholder is 50 years old and his mortality follows that of the illustrative table in \cite{Bowers97}. 

\subsubsection{Risk Measure as Constraint}

The same strategy is used to determine the risk measure level, which means that the statistic $CR$ is obtained for $c$ between 5\% and 95\% and the hedging strategy that provides the smallest value is used as reference. Figure \ref{eia1} gives the distribution of the hedging error when $c=50\%$. The replicating portfolio initial value is 1.00, which is exactly the initial investment from the policyholder. The issuer expects a gain of 1.56\% using this strategy and the expected value of 5\% of the worst losses is only 0.24\%.  

\begin{figure}[h!]
 \includegraphics[scale=0.75] {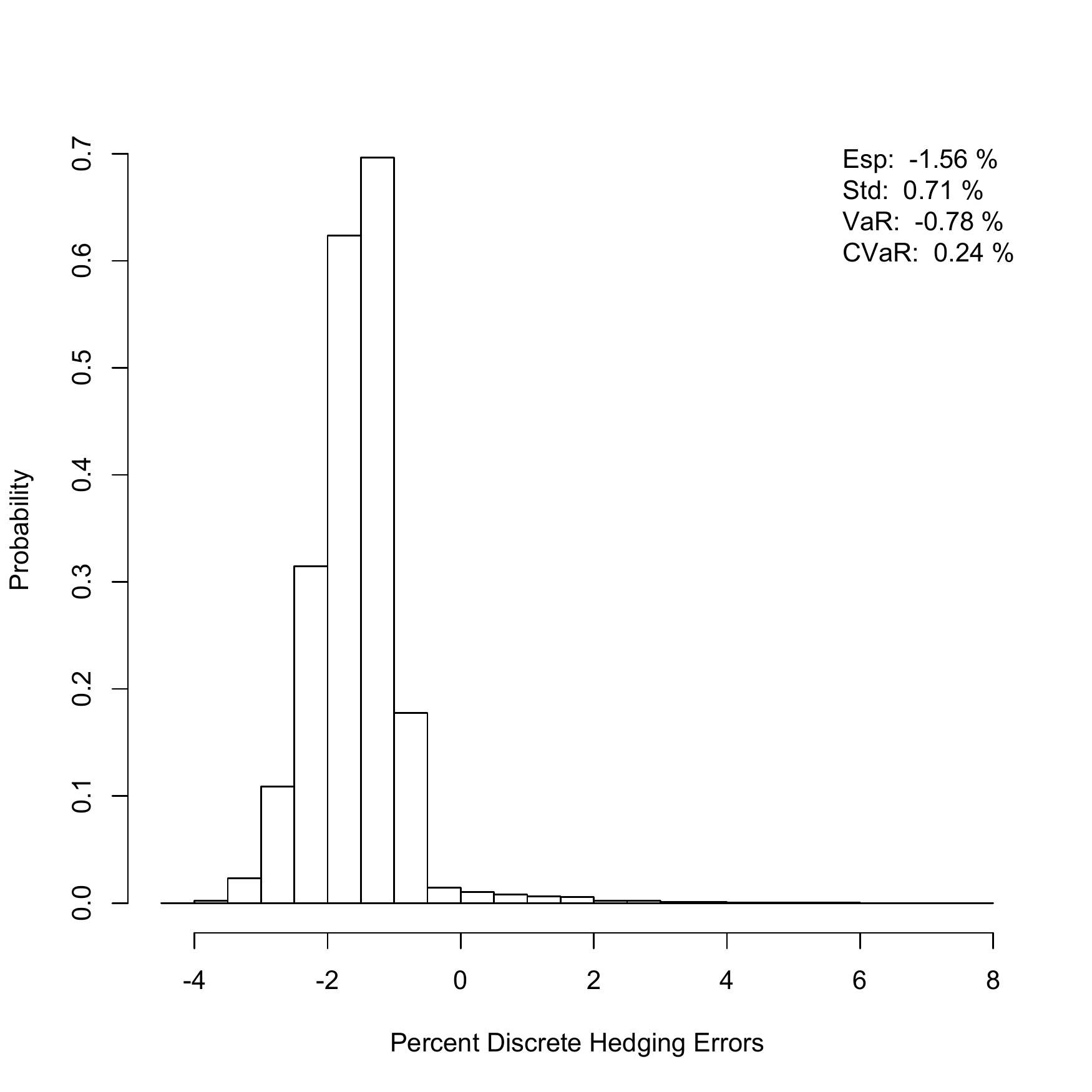}
 \caption{Distribution of the hedging errors for $c=50\%$ using the risk measure as a constraint.}
 \label{eia1}
\end{figure}
   
The issuer could limit the left tail by imposing a limit on the aggregate loss of $\gamma_3=6\%$. The results are displayed in Figure \ref{eiae1}. The hedging statistics are very good even though they are inferior as the ones extracted from Algorithm \ref{algo2} without a state variable. However, the strategy does not admit any large positive deviations.     

\begin{figure}[h!]
 \includegraphics[scale=0.75] {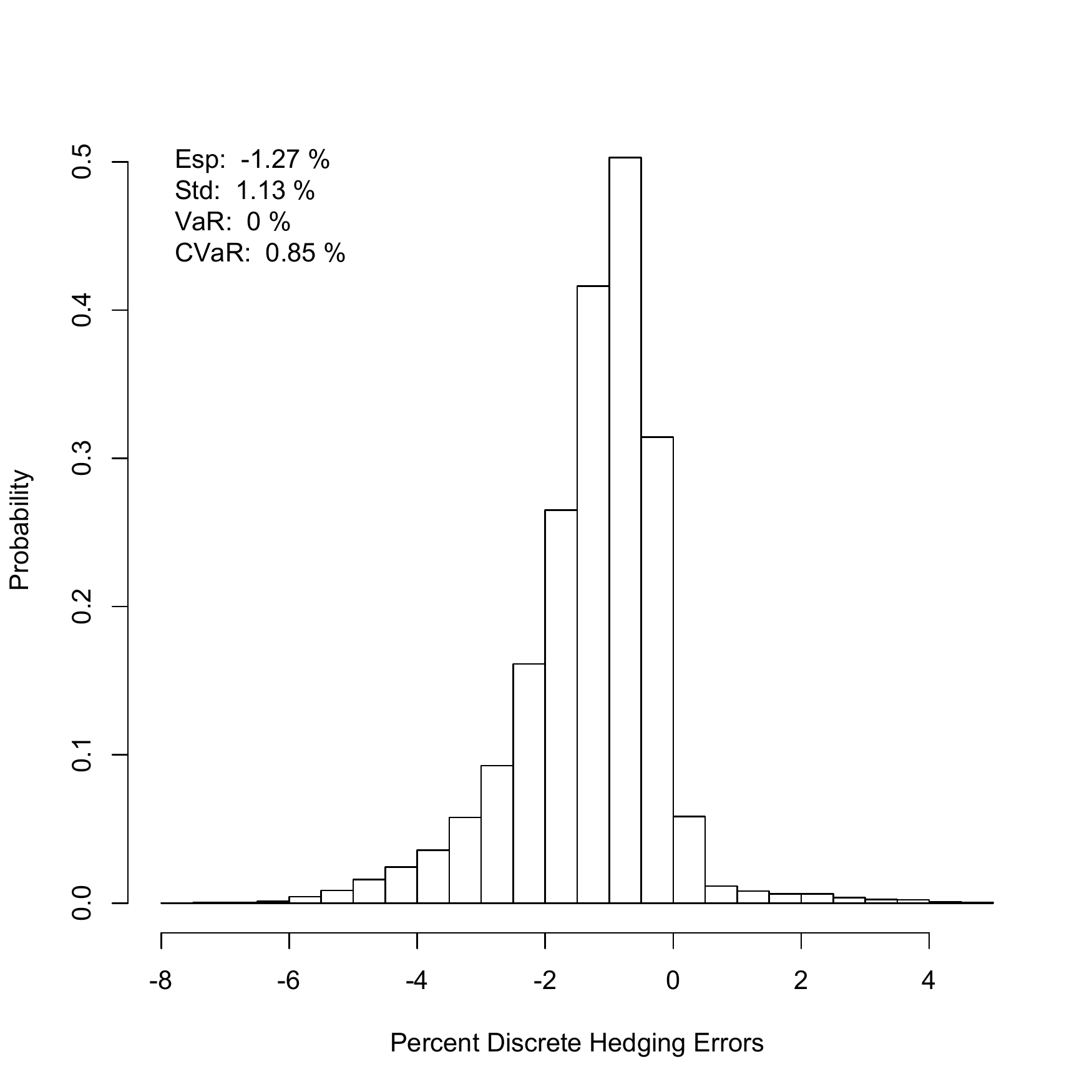}
 \caption{Distribution of the hedging errors for $c=50\%$ using the risk measure as a constraint where $\gamma_3=6\%$.}
 \label{eiae1}
\end{figure}

We also stress the impact of introducing the European option in our hedging strategy for EIAs. Removing the option increases the initial value of the hedging portfolio to 1.01 (compared to 1.00) and the amount $CR$ to 0.85\% (compared to 0.24\%). The impact of introducing the option could be enhance depending on the parameters selections. For example, setting the risk measure threshold to $c=95\%$ results to an initial value for the hedging portfolio of 1.02 with the option and 1.08 without the option.

\subsubsection{Risk Measure as Objective}

Algorithm \ref{algo4} is performed for the EIA contract. Figure \ref{eia3} shows the hedging error distribution where the risk measure is optimized and set to 56\%. The hedging statistics are similar to the ones obtained from Figure \ref{eia1}. Recall that the initial value of the replicating portfolio is one unit. The relative performance of the other algorithms is similar as what it is reported for the GIC analyses.

\begin{figure}[h!]
 \includegraphics[scale=0.75] {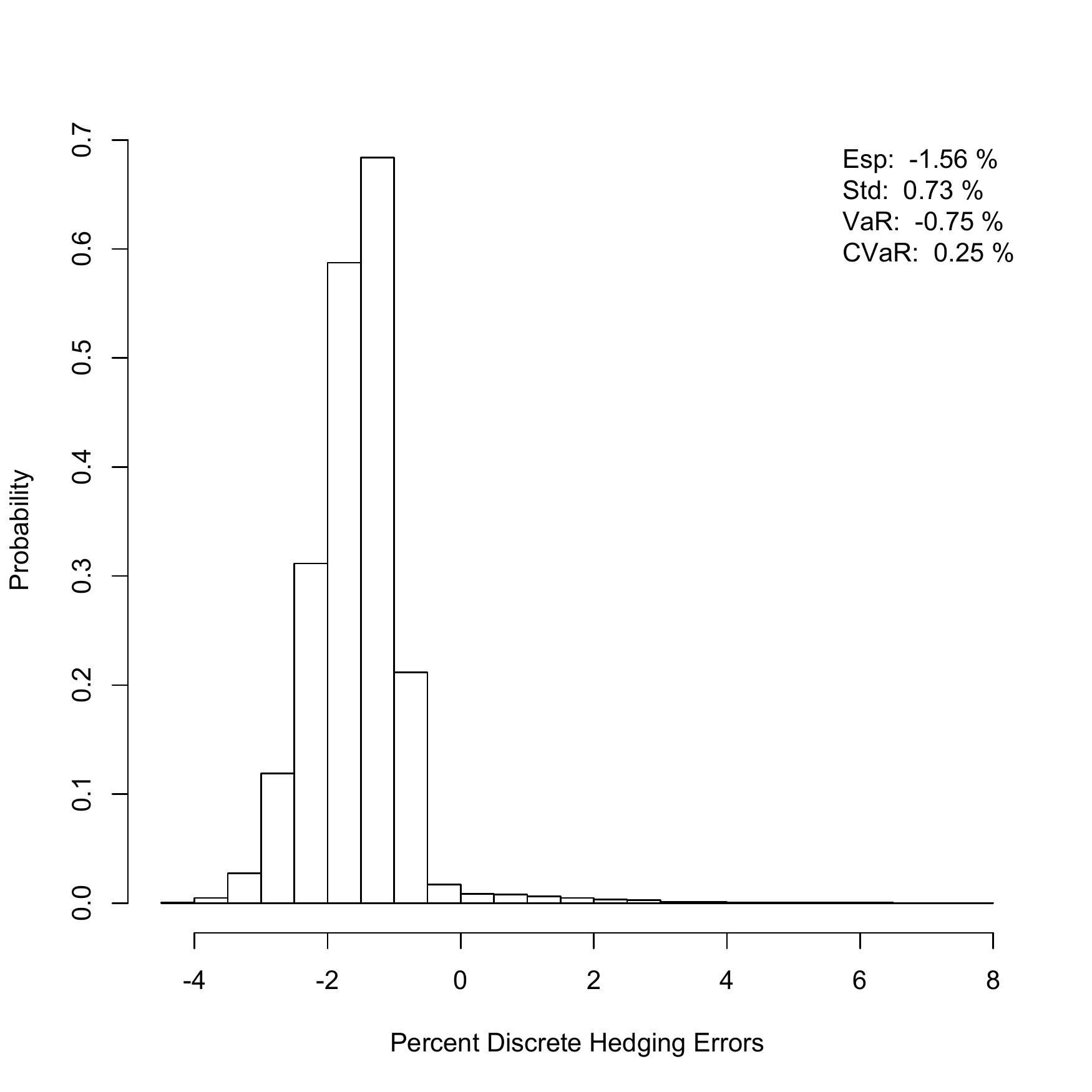}
 \caption{Distribution of the hedging errors for $c=56\%$ under coherent risk measure.}
 \label{eia3}
\end{figure}

\section{Conclusions}

The purpose of this paper is to introduce different risk control models and test their performance on two contingencies: GICs and EIAs.  Our proposed models are based on the control of the discounted loss function. Risk measures are used as instruments to limit the issuer's losses. Two different families of models are presented. On the one hand, the risk measure is used as in the constraints to limit the local risk and on the other hand, the risk measures are used in the objective function. The latter family requires additional local control modelling such as in the second half of Subsection \ref{rmc}. 

An analysis is performed on the distribution of path wise mismatches. A criterion on this distribution is introduced to asses the strategy performance and allows to calibrate the risk measure parameters.

The proposed models and associated algorithms allow the involvement of multiple assets investment.  In our numerical analyses, we allow the issuer to invest in a European call option. This addition shows a great impact on reducing the initial values of the replicating portfolio and improves the hedging strategies.

The algorithms work for recombining trees and unfolded trees where the contingent claim depends on its historic. Hence, the financial framework could be modelled using Markovian or non-Markovian processes.

The tested problems are open to more complex models such as adding transaction costs (Section \ref{hplf}, third paragraph) and or adding more risk control tools as in Subsection \ref{rmc}. Among these tools, the implementation of the control of path wise total losses improves the hedging strategies by limiting the positive range of the error distribution.
 
The special case where the risk measure threshold $c$ is set to 1 leads to the super-replicating strategy. According to this strategy the issuer does not incur any risk. With slightly modified models, other financial derivatives such as American type options can be evaluated. We are testing these modified models and comparing them to benchmark prices.

\vskip0.5cm
\noindent{\Large\bf Acknowledgments}\\
This research was supported by the Natural Sciences and Engineering Research Council of Canada. %The authors are very grateful to the anonymous referees for their valuable comments and suggestions.

\bibliographystyle{plainnat}
\bibliography{bibfile}

 \end{document}